\documentclass[10pt]{llncs}

\title{On Satisficing in Quantitative Games}

\author{Suguman Bansal\inst{1} \Letter 
\and Krishnendu Chatterjee\inst{2}
\and Moshe Y. Vardi\inst{3}}
\authorrunning{S. Bansal, K. Chatterjee, and M.Y. Vardi}
\institute{University of Pennsylvania, Philadelphia, USA \email{suguman@seas.upenn.edu}
\and IST Austria, Klosterneuburg, Austria,
\email{krishnendu.chatterjee@ist.ac.at}
\and Rice University, Houston,  USA \email{vardi@cs.rice.edu}}

\usepackage{algorithm}
\usepackage[noend]{algorithmic}
\usepackage{amsmath}
\usepackage{amsfonts}
\usepackage{amssymb}
\usepackage{color}
\usepackage{comment}
\usepackage{enumitem}
\usepackage{graphicx}
\usepackage[bookmarks=false]{hyperref}
\usepackage{listings}
\usepackage{marvosym}
\usepackage{makeidx}  
\usepackage{subcaption}
\usepackage{url}
\usepackage{wrapfig}

\usepackage[graphicx]{realboxes}
\usepackage{tikz}
\tikzset{elliptic state/.style={draw,ellipse}}
\usetikzlibrary{shapes,shapes.geometric,arrows,fit,calc,positioning,automata,}
\usetikzlibrary{automata,positioning}
\usepackage{tikz-qtree}
\usepackage{float}
\usepackage{floatflt}

\newcommand*{\N}{\mathbb{N}}
\newcommand*{\Q}{\mathbb{Q}}
\newcommand*{\Z}{\mathbb{Z}}
\renewcommand*{\L}{\mathcal{L}}
\renewcommand*{\O}{\mathcal{O}}

\newcommand{\F}{\mathcal{F}}
\newcommand{\ap}{\mathit{AP}}
\newcommand{\GA}{\mathsf{GA}}

\renewcommand*{\mod}[1]{|#1|}

\newcommand*{\floor}[1]{\lfloor #1 \rfloor}

\newcommand*{\A}{\mathcal{A}}
\newcommand*{\Statess}{\mathit{S}}
\newcommand*{\State}{S}
\newcommand*{\Start}{s_I}
\newcommand*{\Final}{\mathcal{F}}

\newcommand*{\wt}{\mathit{wt}}

\newcommand*{\init}{v_\mathsf{init}}

\newcommand*{\pre}[2]{#1[\cdots#2]}
\newcommand*{\post}[2]{#1[#2\cdots]}
\newcommand*{\suf}[2]{#1[#2\dots]}
\newcommand{\gap}[1]{\mathsf{gap}(#1,d)}

\newcommand*{\DSum}[2]{\mathit{DS}({#1}, {#2})}

\newcommand*{\R}{\mathsf{R}}
\newcommand*{\transwt}{\mathsf{cost}}
\newcommand*{\opt}{W}

\newcommand{\optimize}{\mathsf{VIOptimal}}
\newcommand{\compSatisfice}{\mathsf{CompSatisfice}}
\newcommand{\visatisfice}{\mathsf{VISatisfice}}
\begin{document}

\maketitle

%% Abstract

\begin{abstract}

Several problems in planning and reactive synthesis can be reduced to the analysis of two-player quantitative graph games. {\em Optimization} is one form of analysis. We argue that in many cases it may be better to replace the optimization problem with the {\em satisficing problem}, where instead of searching for optimal solutions, the goal is to search for solutions that adhere to a given threshold bound. 

This work defines and investigates the satisficing problem on a two-player graph game with the discounted-sum cost model. We show that while the satisficing problem can be solved using numerical methods just like the optimization problem, this approach does not render compelling benefits over optimization. When the discount factor is, however, an integer, we present another approach to satisficing, which is purely based on automata methods. We show that this approach is algorithmically more performant -- both theoretically and empirically -- and demonstrates the broader applicability of satisficing over optimization.
\end{abstract}

%% Main paper

\section{Introduction}
\label{Sec:Intro}

{\em Quantitative properties} of systems are increasingly being explored in automated reasoning ~\cite{baierprobabilistic,clark2007static,finkbeiner2018model,Kwi07,kwiatkowska2010advances,seshia2018formal}.
In decision-making domains such as planning and reactive synthesis, quantitative properties have been deployed to describe {\em soft constraints} such as quality measures~\cite{bloem2009better},  cost and resources~\cite{he2017reactive,lahijanian2015time}, rewards~\cite{wen2015correct}, and the like. 
Since these constraints are soft, it suffices to generate solutions that are {\em good enough} w.r.t. the quantitative property. 

%MYV: Add citation for the sentence below.
Existing approaches on the analysis of quantitative properties have, however, primarily focused on {\em optimization} of these constraints, i.e., to generate optimal solutions. 
We argue that there may be disadvantages to searching for optimal solutions, where \emph{good enough} ones may suffice. First,  optimization may be more expensive than searching for good-enough solutions. Second, optimization restricts the search-space of possible  solutions, and thus could limit the broader applicability of the resulting solutions. For instance, to generate solutions that operate {\em within} battery life, it is too restrictive to search for solutions with  {\em minimal} battery consumption. Besides, solutions with minimal battery consumption may be limited in their applicability, since they  may not satisfy other goals, such as desirable temporal tasks. 

To this end, this work focuses on directly searching for good-enough solutions. We propose an alternate form of analysis of quantitative properties in which the objective is to search for a solution that adheres to {\em a given threshold bound}, possibly derived from a physical constraint such as battery life. We call this the {\em satisficing problem}, a term popularized by H.A.Simon in economics to mean {\em satisfy and suffice}, implying a search for good-enough solutions~\cite{satisficing}. Through theoretical and empirical investigation, we make the case that satisficing  is algorithmically more performant than optimization and, further, that satisficing solutions may have broader applicability than optimal solutions. 

This work formulates and investigates the satisficing problem on two-player, finite-state games with the discounted-sum (DS) cost model, which is a standard cost-model in decision-making domains~\cite{osborne1994course,puterman1990markov,sutton1998introduction}. 
In these games, players take turns to pass a token along the {\em transition relation} between the states. As the token is pushed around, the play accumulates costs along the transitions using the  DS cost model. The players are assumed to have opposing objectives: one player maximizes the cost, while the other player minimizes it. 
We define the satisficing problem as follows:  {\em Given a {\em threshold value $v \in \Q$}, does there exist a strategy for the minimizing (or maximizing) player that ensures the cost of all resulting plays is strictly or non-strictly lower (or greater) than the threshold $v$?} 

Clearly, the satisficing problem is decidable since the optimization problem on these quantitative games is known to be solvable in pseudo-polynomial time~\cite{hansen2013strategy,littman1996algorithms,zwick1996complexity}. 
To design an algorithm for satisficing, we first adapt the celebrated value-iteration (VI) based algorithm for optimization~\cite{zwick1996complexity} (\textsection~\ref{Sec:Optimization}). We show, however, that this algorithm, called $\visatisfice$, displays the same complexity as optimization and hence renders no complexity-theoretic advantage. To obtain worst-case complexity, we perform a thorough worst-case analysis of VI for optimization. It is interesting that a thorough  analysis of VI for optimization had hitherto been absent from the literature, despite the popularity of VI. To address this gap, we first prove that VI should be executed for $\Theta(|V|)$ iterations to compute the optimal value, where $V$ and $E$ refer to the sets of states and transitions in the quantitative game. Next, to compute the overall complexity, we take into account the cost of arithmetic operations as well, since they appear in abundance in VI. We demonstrate an orders-of-magnitude difference between the complexity of VI  under different cost-models of arithmetic. For instance, for integer discount factors, we show that VI  is %$\O(|V|^2\cdot |E|)$ and $\O(|V|^4\cdot |E|)$ 
$\O(|V|\cdot |E|)$ and $\O(|V|^2\cdot |E|)$ under the unit-cost and bit-cost models of arithmetic, respectively. Clearly, this shows that VI for optimization, and hence $\visatisfice$, does not scale to large quantitative games. 

We then present a  purely automata-based approach for satisficing (\textsection~\ref{Sec:Satisfice}). While this approach applies to integer discount factors only, it solves satisficing in $\O(|V|+|E|)$ time. This shows that there is a fundamental separation in complexity between satisficing and VI-based optimization, as even the lower bound on the number of iterations in VI is higher. In this approach, the satisficing problem is reduced to solving a safety or reachability game. Our core observation is that the criteria to fulfil  satisficing with respect to threshold value $v\in\Q$ can be expressed 
%MYV: Noot clear what is a member of what.
as membership in an automaton 
that accepts a weight sequence $A$ iff $\DSum{A}{d}$ $\mathsf{R}$ $v$ holds, where $d>1$ is the discount factor and $\mathsf{R} \in \{\leq, \geq, <,>\}$. In existing literature, such automata are called {\em comparator automata} (comparators, in short) when the threshold value  $v=0$~\cite{BCVFoSSaCS18,BCVlmcs2019}. They are known to have a compact safety or co-safety automaton representation~\cite{BVCAV19,kupferman1999model}, which could be used to reduce the satisficing problem with zero threshold value. To solve satisficing for arbitrary threshold values $v\in Q$, we extend existing results on comparators to permit arbitrary but fixed threshold values $ v \in \Q$. An empirical comparison between the performance of $\visatisfice$, VI for optimization, and automata-based solution for satisficing shows that the latter outperforms the others in efficiency, scalability, and robustness.  

In addition to improved algorithmic performance, we demonstrate that satisficing solutions have broader applicability than optimal ones (\textsection~\ref{Sec:Temporalgoals}). 
We examine this with respect to their ability to extend 
to temporal goals. That is, the problem is to find optimal/satisficing solutions that also satisfy a given temporal goal. Prior results have shown this to not be possible with optimal solutions~\cite{chatterjee2017quantitative}. In contrast, we show satisficing extends to temporal goals when the discount factor is an integer. This occurs because both satisficing and satisfaction of temporal goals are solved via automata-based techniques, which can be easily integrated.  

In summary, this work contributes to showing that satisficing has algorithmic and applicability advantages over optimization in (deterministic) quantitative games. In particular, we have shown that the automata-based approach for satisficing have advantages over approaches in numerical methods like value-iteration. This gives yet another evidence in favor of automata-based quantitative reasoning and opens up several compelling directions for future work. 

% We show that this approach is algorithmically more performant -- both theoretically and empirically -- and demonstrates the broader applicability of satisficing over optimization.
 
%A long version of this paper with complete proofs can be found at~\cite{tacas2021}.

\section{Preliminaries}
\label{Sec:Prelims}

\subsection{Two-player graph games}

\paragraph{Reachability and safety games.}

Both {\em reachability} and {\em safety games} are defined over the structure  $G = (V  = V_0 \uplus V_1, \init, E, \F)$~\cite{thomas2002automata}.
It consists of a directed graph $(V,E)$,
and a partition $(V_0, V_1)$ of its  states $V$. 
State $\init$ is the {\em initial state} of the game. 
The set of successors of state $v$ is designated by $vE $. For convenience, we assume that every state has at least one outgoing edge, i.e, $vE  \neq \emptyset$ for all $v \in V$. $ \F \subseteq V$ is a non-empty set of states. $\F$ is referred to as {\em accepting} and {\em rejecting} states in reachability and safety games, respectively. 

A {\em play} of a game involves two players, denoted by $P_0$ and $P_1$, to create an infinite path by moving a token along the transitions as follows:
At the beginning, the token is at the initial state. If the current
position $v$ belongs to $V_i$, then $P_i$ chooses the successor state from 
$vE$.
Formally,
a play $\rho = v_0v_1v_2\dots$ is an infinite sequence of states such that the first state $v_0 = v_{\mathsf{init}}$, and each pair of successive states is a transition, i.e., $(v_k, v_{k+1}) \in E$ for all $k\geq 0$. 
A play is  {\em winning for player $P_1$} in a reachability game if it visits an accepting state, and {\em winning for player $P_0$} otherwise. 
The opposite holds in safety games, i.e., a play is winning for player $P_1$ if it does not visit any rejecting state, and winning for $P_0$ otherwise.

 A {\em strategy} for a player is a
recipe that guides the player on which state to go next to based on the history of the play.  
% is a partial function $\Pi_i : V^*V_i\rightarrow V$ such that $v = \Pi_i(v_0v_1\dots v_k)$ belongs to $v_kE$ . 
 %A player $P_i$ is said to follow a strategy $\Pi$ on a play $\rho$ if for all $k$-length prefixes $v_0v_1\dots v_{k-1}$ of $\rho$ for all $v_{k-1} \in V_i$ it is that $v_k = \Pi(v_0v_1\dots v_{k-1})$.
A {\em strategy is  winning for a player $P_i$} if for all strategies of the opponent player $P_{1-i}$, the resulting plays are winning for $P_i$. 
 To {\em solve} a graph game means to  determine whether there exists a winning strategy for player $P_1$. 
Reachability and safety games are solved in $\O(|V|+|E|)$.

\paragraph{Quantitative graph games.}

A {\em quantitative graph game} (or quantitative game, in short) is defined over a structure $G = (V  = V_0 \uplus V_1, \init, E, \gamma)$. $V$, $V_0$, $V_1$, $\init$, $E$, plays and strategies are defined as earlier.  
Each transition of the game is associated with a {\em cost} determined by the {\em cost function} $\gamma: E \rightarrow \Z$.
The {\em cost sequence} of a play $\rho$ is the sequence of costs $w_0w_1w_2\dots$ such that  $w_k = \gamma((v_k, v_{k+1}))$ for all $i\geq 0$.
Given a discount factor $d>1$, the {\em cost of play} $\rho$, denoted $\wt(\rho)$, is the discounted sum of its cost sequence, i.e., $\wt(\rho) = \DSum{\rho}{d} = w_0 + \frac{w_1}{d} + \frac{w_2}{d^2} + \dots$.

\subsection{Automata and formal languages}

\paragraph{B\"uchi automata.}
A  {\em B\"uchi automaton} is a tuple
  $\A = (\Statess$, $\Sigma$, $\delta$, $s_\mathcal{I}$, $\Final)$, where
  $ \Statess $ is a finite set of {\em states}, $ \Sigma $ is a finite {\em input alphabet},  $ \delta \subseteq (\Statess \times \Sigma \times \Statess) $ is the   {\em transition relation}, state $ s_\mathcal{I} \in \Statess $ is the {\em initial state}, and $ \Final \subseteq \Statess $ is the set of {\em accepting states}~\cite{thomas2002automata}.
A B\"uchi automaton is {\em deterministic} if for all states $ s $ and
inputs $a$, $ |\{s'|(s, a, s') \in \delta \textrm{ for some $s'$} \}|
\leq 1 $.
%and $|\StartState|=1$. 
%Otherwise, it is {\em nondeterministic}.
%A B\"uchi automaton is {\em complete} if for all states $ s $ andinputs $a$, $ |\{s'|(s, a, s') \in \delta \textrm{ for some $s'$} \}| \geq 1$. 
For a word $ w = w_0w_1\dots \in \Sigma^{\omega} $, a {\em run} $ \rho$ of $ w $ is a sequence of states $s_0s_1\dots$ s.t.
$ s_0 = s_\mathcal{I}$, and $ \tau_i =(s_i, w_i, s_{i+1}) \in \delta $ for all $i$.
Let $ \mathit{inf}(\rho) $ denote the set of states that occur infinitely
often in run ${\rho}$. 
A run $\rho$ is an {\em accepting run} if $ \mathit{inf}(\rho)\cap \Final \neq \emptyset $. A word $w$ is an accepting word if it has an accepting run. 
The language of B\"uchi automaton $\A$ is the set of all words accepted by $\A$. 
Languages accepted by B\"uchi automata are called {\em $\omega$-regular}.
%By abuse of notation, we write $ w \in \A $ and $\rho \in \A$ if $w$ and $\rho$ are an accepting word and an accepting run of $\A$.
%B\"uchi automata are  closed under set-theoretic union, intersection, and complementation~\cite{thomas2002automata}. 

%We review {\em safety and co-safety languages}.  
\paragraph{Safety and co-safety languages.}
Let $ \L \subseteq \Sigma^\omega$ be a language  over alphabet $\Sigma$. 
A finite word $w \in \Sigma^*$ is a {\em bad prefix} for $\L$  if for all infinite words $y \in \Sigma^\omega$, $x\cdot y \notin \L$.
A language $\L$ is a {\em safety language} if every word $ w \notin \L$ has a bad prefix for $\L$~\cite{alpern1987recognizing}.
A {\em co-safety language} is the complement of a safety language~\cite{kupferman1999model}.
Safety and co-safety languages that are $\omega$-regular are represented by specialized B\"uchi automata called {\em safety} and {\em co-safety automata}, respectively.

\paragraph{Comparison language and comparator automata.}

Given integer bound $\mu>0$, discount factor $d >1$, and relation $\mathsf{R} \in \{<,>,\leq, \geq, =, \neq\}$
the {\em comparison language with upper bound $\mu$, relation $\mathsf{R}$,  discount factor $d$} is the language of  words over the alphabet $\Sigma = \{-\mu, \dots, \mu\}$ that accepts  $A \in \Sigma^{\omega}$ iff $\DSum{A}{d}$ $\mathsf{R}$ $0$ holds~\cite{BCVCAV18,BVCAV19}. The {\em comparator automata  with upper bound $\mu$, relation $\mathsf{R}$, discount factor $d$} is the automaton that accepts the corresponding comparison language~\cite{BCVFoSSaCS18}. 
Depending on $\R$, these languages are safety or co-safety~\cite{BVCAV19}.
A comparison language is said to be {\em $\omega$-regular} if its automaton is a B\"uchi automaton.  Comparison languages are $\omega$-regular iff the discount factor is an integer~\cite{BCVlmcs2019}.

\section{Satisficing via Optimization}
\label{Sec:Optimization}

This section shows that there are no complexity-theoretic benefits to solving the satisficing problem via algorithms for the optimization problem. 

\textsection~\ref{Sec:Review} formally defines the satisficing problem and reviews the celebrated value-iteration (VI) algorithm for optimization by Zwick and Patterson (ZP). While ZP {\em claim without proof} that the algorithm runs in pseudo-polynomial time~\cite{zwick1996complexity}, its worst-case analysis is absent from literature.  This section presents a detailed account of the said analysis, and exposes the dependence of VI's worst-case complexity on the discount factor $d>1$ and the cost-model for arithmetic operations i.e. unit-cost or bit-cost model. The analysis is split into two parts: First,  \textsection~\ref{Sec:NumIteration}  shows it is sufficient to terminate after a finite-number of iterations. Next, \textsection~\ref{Sec:worstcasecomplexityVI} accounts for the cost of  arithmetic operations per iteration to compute VI's worst-case complexity under unit- and bit-cost cost models of arithmetic 
Finally, \textsection~\ref{Sec:VIForSatisfice} presents and analyzes our VI-based algorithm for satisficing $\visatisfice$.

\subsection{Satisficing and Optimization}
\label{Sec:Review}

\begin{definition}[Satisficing problem]
Given a quantitative graph game $G$ and a threshold value $v \in \Q$, the {\em satisficing problem} is to determine whether the minimizing (or maximizing) player has a strategy that ensures the cost of all resulting plays is strictly or non-strictly lower (or greater)  than the threshold $v$.
\end{definition}

The satisficing problem can clealy be solved by solving the \emph{optimization problem}. The optimal cost of a quantitative game is that value such that the maximizing and minimizing players can guarantee that the cost of plays is at least and at most the optimal value, respectively. 

\begin{definition}[Optimization problem]
Given a quantitative graph game $G$, the {\em optimization problem} is to compute the optimal cost from all possible plays from the game, under the assumption that the players have opposing objectives to maximize and minimize the cost of plays, respectively. 
\end{definition}

Seminal work by Zwick and Patterson showed the optimization problem is solved by the value-iteration algorithm presented here~\cite{zwick1996complexity}. Essentially, the algorithm plays a  min-max game between the two players. 
Let $\wt_k(v)$ denote the optimal cost of a $k$-length game that begins in state $v \in V$.
Then $\wt_k(v)$ can be computed using the following equations:
The optimal cost of a 1-length game beginning in  state $v \in V$ is  $\max\{ \gamma(v, w) | (v,w)\in E\}$ if  $v \in V_0$ and  $\min\{ \gamma(v, w) | (v,w)\in E\}$ if  $v \in V_1$.  
\begin{comment}
\begin{align*}
    \wt_1(v) =
    \begin{cases}
        \textit{max}\{ \gamma(v, w) | (v,w)\in E\} &\text{ if } v \in V_0 \\
        \textit{min}\{ \gamma(v, w) | (v,w)\in E\} &\text{ if } v \in V_1
    \end{cases}
\end{align*}
\end{comment}
Given the optimal-cost of a $k$-length game, the optimal cost of a $(k+1)$-length game is computed as follows:
\begin{align*}
    \wt_{k+1}(v) = 
    \begin{cases}
    \textit{max}\{ \gamma(v, w) + \frac{1}{d}\cdot \wt_k(w) | (v,w) \in E\} \text{ if } v \in V_0 \\
    \textit{min}\{ \gamma(v, w) + \frac{1}{d}\cdot \wt_k(w) | (v,w) \in E\} \text{ if } v \in V_1 
    \end{cases}
\end{align*}

Let $\opt$ be the optimal cost. Then,
$\opt = \lim_{k\rightarrow \infty}\wt_k(\init)$.~\cite{shapley1953stochastic,zwick1996complexity}.
 
 \subsection{VI: Number of iterations}
\label{Sec:NumIteration}

The VI algorithm described above terminates at {\em infinitum}. To compute the algorithms' worst-case complexity, we establish a linear bound on the number of iterations that is sufficient to compute the optimal cost. We also establish a matching lower bound, showing that our analysis is tight. 

\paragraph{Upper bound on number of iterations.}
The upper bound computation utilizes one key result from existing literature: There exist memoryless strategies for both players such that the cost of the resulting play is the optimal cost~\cite{shapley1953stochastic}. Then, there must exists an optimal play in the form of a {\em simple lasso} in the quantitative game, where a {\em lasso} is a play represented as $v_0v_1\dots v_n(s_0 s_2\dots s_m)^\omega$. We call the initial segment $v_0v_1\dots v_n$ its {\em head}, and the cycle segment $s_0s_1\dots s_m$ its {\em loop}.
A lasso is {\em simple} if each state in $\{v_0\dots v_n, s_0,\dots s_m\}$ is distinct. We begin our proof by assigning constraints on the optimal cost using the simple lasso structure of an optimal play (Corollary~\ref{coro:optimalplay} and Corollary~\ref{coro:optimaldiff}). 

Let $l = a_0\dots a_n(b_0\dots b_m)^{\omega}$ be the cost sequence of a lasso such that $l_1 = a_0\dots a_n$ and  $l_2 = b_0\dots b_m$ are the cost sequences of the head and the loop, respectively. Then the following can be said about $\DSum{l_1\cdot l_2^\omega}{d}$,
\begin{lemma}
\label{lem:optimalplay}
Let $l = l_1\cdot (l_2)^{\omega}$ represent  an integer cost sequence of a lasso, where $l_1$ and $l_2$ are the cost sequences of the head and loop of the lasso. Let $d = \frac{p}{q}$ be the discount factor. 
Then, $\DSum{l}{d}$ is a rational number with denominator at most ($p^{|l_2|} - q^{|l_2|})\cdot (p^{|l_1|})$.
\end{lemma}

Lemma~\ref{lem:optimalplay} is proven by unrolling $\DSum{l_1\cdot l_2^\omega}{d}$. 
%The full proof is in the supplemental material. 
Then, the first constraint on the optimal cost is as follows:

\begin{corollary}
\label{coro:optimalplay}
Let $G=(V, \init, E, \gamma)$ be a quantitative graph game. Let $d = \frac{p}{q}$ be the discount factor. Then the optimal cost of the game is a rational number with denominator at most  $(p^{|V|} - q^{|V|})\cdot (p^{|V|})$ 
\end{corollary}
\begin{proof}
Recall, there exists a simple lasso that computes the optimal cost. Since a simple lasso is of $|V|$-length at most, the length of its head and loop are at most $|V|$  each. So, the expression from Lemma~\ref{lem:optimalplay} simplifies to   $(p^{|V|} - q^{|V|})\cdot (p^{|V|})$.
\qed
\end{proof}

The second constraint has to do with the minimum non-zero difference between the cost of simple lassos:
\begin{corollary}
\label{coro:optimaldiff}
Let $G=(V, \init, E, \gamma)$ be a quantitative graph game. Let $d = \frac{p}{q}$ be the discount factor.
Then the  minimal non-zero difference between the cost of simple lassos is a rational with denominator at most 
%$(p^{(|V|^2)} - q^{(|V|^2)})\cdot (p^{(|V|^2)})$.
$(p^{(|V|)} - q^{(|V|)})^2\cdot (p^{(2\cdot|V|)})$.
\end{corollary}
\begin{proof}

Given two rational numbers with denominator at most $a$, an upper bound on the denominator of minimal non-zero difference of these two rational numbers is $a^2$. 
Then, using the result from Corollary~\ref{coro:optimalplay}, we immediately obtain that the minimal non-zero difference  between the cost of two lassos is 
a rational number with denominator at most $(p^{(|V|)} - q^{(|V|)})^2\cdot (p^{(2\cdot|V|)})$.
%The difference of two simple lassos $l_1$ and $l_2$ of length at most $|V|$ can be represented  by another simple lasso $l = l_1\times l_2$ of length at most $|V|^2$. If the maximum length of the lassos is $|V|$, then the maximum length of the difference lasso will be $|V|^2$. Then, from Lemma~\ref{lem:optimalplay} we immediately obtain that the maximum value of the denominator of the minimum non-zero difference of simple lassos is $(p^{(|V|^2)} - q^{(|V|^2)})\cdot (p^{(|V|^2)})$.
\qed
\end{proof}

For notational convenience, let $\mathsf{bound}_{{\opt}} = (p^{|V|} - q^{|V|})\cdot (p^{|V|})$ and
%$\mathsf{bound}_{\mathsf{diff}} = (p^{(|V|^2)} - q^{(|V|^2)})\cdot (p^{(|V|^2)})$.
$\mathsf{bound}_{\mathsf{diff}} = (p^{(|V|)} - q^{(|V|)})^2\cdot (p^{(2\cdot |V|)})$.
Wlog, $|V|>1$. 
Since, $\frac{1}{\mathsf{bound}_{\mathsf{diff}}} < \frac{1}{\mathsf{bound}_{{\opt}}}$, there is at most one rational number with denominator $\mathsf{bound}_{\mathsf{W}}$ or less in any interval of size
$\frac{1}{\mathsf{bound}_{\mathsf{diff}}}$. 
Thus, if we can identify an interval of size less than $\frac{1}{\mathsf{bound}_{\mathsf{diff}}}$ around the optimal cost, then due to Corollary~\ref{coro:optimalplay}, the optimal cost will be the unique rational number with denominator $\mathsf{bound}_{\mathsf{W}}$ or less in this interval. 

Thus, the final question is to identify a small enough interval (of size $\frac{1}{\mathsf{bound}_{\mathsf{diff}}}$ or less) such that the optimal cost lies within it. To find an interval around the optimal cost, we use
a finite-horizon approximation of the optimal cost:
\begin{lemma}
\label{lem:interval}
Let $\opt$ be the optimal cost in quantitative game $G$. Let $\mu>0$ be the maximum of absolute value of cost on  transitions in $G$. Then, for all $k \in \N$, 
\[
    \wt_k(v_{\mathsf{init}}) - \frac{1}{d^{k-1}}\cdot \frac{\mu}{d-1} \leq
    \opt \leq
    \wt_k(v_{\mathsf{init}}) + \frac{1}{d^{k-1}}\cdot \frac{\mu}{d-1}
    \]
\end{lemma}
\begin{proof}
    Since $\opt$ is the limit of $\wt_k(v_{\mathsf{init}})$ as $k\rightarrow \infty$, $\opt$  must lie in between the
    minimum and maximum cost possible if the $k$-length game is extended to an infinite-length game. 
    The minimum possible extension would be when the $k$-length game is extended by iterations in which the cost incurred in each round is $-\mu$. Therefore, the minimum possible value is $\wt_k(v_{\mathsf{init}}) - \frac{1}{d^{k-1}}\cdot \frac{\mu}{d-1} $.
    Similarly, the maximum possible value is $\wt_k(v_{\mathsf{init}}) + \frac{1}{d^{k-1}}\cdot \frac{\mu}{d-1} $.
\qed
\end{proof}

Now that we have an interval around the optimal cost, we can compute the number of iterations of VI required to make it smaller than $1/\mathsf{bound}_{\mathsf{diff}}$. \begin{theorem}
\label{thrm:vsquareiterations}
Let $G = (V, \init,E,\gamma)$ be a quantitative graph game.
Let $\mu>0$ be the maximum of absolute value of costs along transitions. The number of iterations required by the value-iteration algorithm  is
\begin{enumerate}
    \item $\O(|V|)$ when discount factor $d\geq 2$, 
    \item $\O\Big(\frac{\log(\mu)}{d-1} + |V|\Big)$ when discount factor $1<d<2$.
    %\item $\O(|V|^2)$ when discount factor $d\geq 2$, 
    %\item $\O\Big(\frac{\log(\mu)}{d-1} + |V|^2\Big)$ when discount factor $1<d<2$.
\end{enumerate}
\end{theorem}
\begin{proof}[Sketch]
As discussed in Corollary~\ref{coro:optimalplay}-\ref{coro:optimaldiff} and Lemma~\ref{lem:interval}, the optimal cost is the unique rational number with denominator $\frac{1}{\mathsf{bound}_{W}}$  or less within the interval $(\wt_k(v_{\mathsf{init}}) - \frac{1}{d^{k-1}}\cdot \frac{\mu}{d-1} , \wt_k(v_{\mathsf{init}}) + \frac{1}{d^{k-1}}\cdot \frac{\mu}{d-1} )$ for a large enough $k>0$ such that the interval's size is less than $\frac{1}{\mathsf{bound}_{\mathsf{diff}}}$. Thus, our task is to determine the value of $k>0$ such that
$
2\cdot\frac{\mu}{d-1\cdot d^{k-1}} \leq \frac{1}{\mathsf{bound}_{\mathsf{diff}}} 
$ holds. The case $d\geq 2$ is easy to simplify.  The case $1<d<2$ involves approximations of logarithms of small values. 
%Details in supplemental material. 
\qed 
\end{proof}

\paragraph{Lower bound on number of iterations of VI.}

%\begin{wrapfigure}{r}{\textwidth}
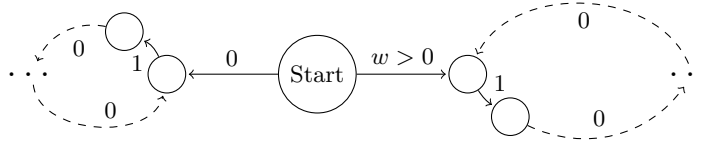
\begin{figure}[t]
\vspace{-1cm}
    \centering
    
    \begin{tikzpicture}[shorten >=1pt,on grid,auto] 
   
    \node[state] (q1) {\footnotesize{Start}};
    \node[state] (q2)  [right of = q1, node distance=2cm, minimum size=0.5cm] {};
    \node[state] (q3)  [left of = q1, node distance=2cm, minimum size=0.5cm] {};
    \node[] (q4)  [right of = q1, node distance=5cm] {\Large{\textellipsis}};
    \node[] (q5) [left of=q3, node distance=1.8cm] {\Large\textellipsis};
    \node[state] (q6) [below right of = q2, minimum size=0.5cm, node distance=0.8cm] {};
    \node[state] (q7) [above left of = q3, minimum size=0.5cm, node distance=0.8cm] {};
  
    \path[->] 
    (q1) edge  node {$w>0$} (q2)
          edge  node [above] {$0$} (q3)
         
    (q2) edge [bend right = 12] node {1} (q6) 
    (q6) edge [bend right = 40, dashed] node   {0} (q4) 
    (q3) edge [bend right = 20] node {1} (q7) 
    (q7) edge [bend right = 40, dashed] node   {0} (q5) 
    (q4) edge [bend right = 65, dashed] node  {0} (q2)
    (q5) edge [bend right = 80, dashed] node {0} (q3);
    \end{tikzpicture}
    \caption{Sketch of game graph which requires $\Omega(|V|)$ iterations}
\label{fig:lowerboundexample}
%\end{wrapfigure} 
\end{figure}
We establish a matching lower bound of $\Omega(|V|)$ iterations to show that our analysis is tight. 

Consider the sketch of a quantitative game in Fig~\ref{fig:lowerboundexample}. 
Let all states belong to the maximizing player. Hence, the optimization problem reduces to searching for a {\em path} with optimal cost. 
Now let the loop on the right-hand side (RHS) be larger than the loop on the left-hand side (LHS).  
For carefully chosen values of $w$ and lengths of the loops, one can show that the path for optimal cost of a $k$-length game is along the RHS loop when $k$ is small, but along the  LHS loop when $k$ is large. 
This way, the correct maximal value can be obtained only at a large value for $k$. Hence the VI algorithm runs for at least enough iterations that the optimal path will be in the LHS loop. 
By meticulous reverse engineering of the size of both loops and the value of $w$, one can guarantee that $k = \Omega(|V|)$.
%As a concrete instance, let the RHS loop and LHS loops have $4n$ and $2n$ edges, respectively, and $w = \frac{1}{d^{3n}} + \frac{1}{d^{7n}} + \cdots + \frac{1}{d^{m\cdot n -1}}$ for a constant value $m>0$

\subsection{Worst-case complexity analysis of VI for optimization}
\label{Sec:worstcasecomplexityVI}

Finally, we complete the worst-case complexity analysis of VI for optimization. We account for the the cost of arithmetic operations since they appear in abundance in VI. 
We demonstrate that there are orders-of-magnitude of difference in complexity under different models of arithmetic, namely unit-cost and bit-cost.

\paragraph{Unit-cost model.}
Under the unit-cost model of arithmetic, all arithmetic operations are assumed to take constant time. 
\begin{theorem}
\label{thrm:optimizationunitcost}
Let $G = (V, \init, E, \gamma)$ be a quantitative graph game. Let $\mu>0$ be the maximum of absolute value of costs along transitions.
The worst-case complexity of the optimization problem under unit-cost model of arithmetic is  
\begin{enumerate}
     \item $\O(|V|\cdot |E|)$ when discount factor $d\geq 2$, 
    \item $\O\Big(\frac{\log(\mu)\cdot |E|}{d-1} + |V|\cdot |E|\Big)$ when discount factor $1<d<2$.
%     \item $\O(|V|^2\cdot |E|)$ when discount factor $d\geq 2$, 
%    \item $\O\Big(\frac{\log(\mu)\cdot |E|}{d-1} + |V|^2\cdot |E|\Big)$ when discount factor $1<d<2$.
\end{enumerate}
\end{theorem}
\begin{proof}
Each iteration takes $\O(E)$ cost since every transition is visited once. 
Thus, the complexity is $\O(|E|)$ multiplied by the number of iterations (Theorem~\ref{thrm:vsquareiterations}). 
\qed
\end{proof}

\paragraph{Bit-cost model.}
Under the bit-cost model, the cost of arithmetic operations depends on the size of the numerical values. Integers are represented in their bit-wise representation. Rational numbers $\frac{r}{s}$ are represented as a tuple of the bit-wise representation of integers $r$ and $s$. For two integers of length $n$ and $m$, the cost of their addition and multiplication is $O(m+n)$ and $O(m\cdot n)$, respectively.

\begin{theorem}
\label{thrm:optimizationbitcost}
Let $G = (V, \init, E, \gamma)$ be a quantitative graph game. 
Let $\mu>0$ be the maximum of absolute value of costs along transitions. Let $d = \frac{p}{q}>1$ be the discount factor.
The worst-case complexity of the optimization problem under the bit-cost model of arithmetic is
\begin{enumerate}
    \item $\O(|V|^2 \cdot |E| \cdot \log p \cdot \max\{\log \mu, \log p\})$ when $d \geq 2$, 
    \item $\O\Big( \Big( \frac{\log(\mu)}{d-1} + |V|\Big)^2 \cdot |E| \cdot \log p \cdot \max\{\log \mu, \log p\}\Big)$ when $1<d < 2$.
    
%    \item $\O(|V|^4 \cdot |E| \cdot \log p \cdot \max\{\log \mu, \log p\})$ when $d \geq 2$, 
%   \item $\O(\Big(\frac{\log(\mu)}{d-1} + |V|^2\Big)^2 \cdot |E| \cdot \log p \cdot \max\{\log \mu, \log p\})$ when $1<d < 2$.
\end{enumerate}
\end{theorem}
\begin{proof}[Sketch]
Since arithmetic operations incur a cost and the length of representation of intermediate costs increases linearly in each iteration, we can show that  the cost of conducting the $j$-th iteration is $\O(|E|\cdot j \cdot \log \mu \cdot \log p)$. Their summation will return the given expressions.
\qed
\end{proof}

\paragraph{Remarks on integer discount factor.}
Our analysis shows that when the discount factor is an integer ($d\geq 2 $), VI requires $\Theta(|V|)$ iterations. Its worst-case complexity is, therefore, $\O(|V|\cdot |E|)$ and $\O(|V|^2\cdot |E|)$ under the unit-cost and bit-cost models for arithmetic, respectively.
From a practical point of view, the bit-cost model is more relevant since implementations of VI will use multi-precision libraries to avoid floating-point errors. While one may argue that the upper bounds in Theorem~\ref{thrm:optimizationbitcost} could be tightened, they would not improve significantly due to the $\Omega(|V|)$ lower bound on number of iterations.

\subsection{Satisficing via value-iteration}
\label{Sec:VIForSatisfice}

We present our first algorithm for the satisficing problem. It is an adaptation of VI. However, we see that it does not fare better than VI for optimization.

VI-based algorithm for satisficing is described as follows: Perform VI for optimization. Terminate as soon as one of these  occurs: (a). VI completes as many iterations from Theorem~\ref{thrm:vsquareiterations}, or (b). The threshold value falls
outside the interval defined in Lemma~\ref{lem:interval}. Either way, one can tell how the threshold value relates to the optimal cost to solve satisficing. 
Clearly, (a) needs as many iterations as optimization; (b) does not reduce
the number of iterations since it is inversely proportional to the distance between optimal cost and threshold value:
\begin{theorem}
\label{thrm:nonrobust}
Let $G = (V, \init, E, \gamma)$ be a quantitative graph game with optimal cost $\opt$. Let $v \in \Q$ be the threshold value. 
Then number of iterations taken by a VI-based algorithm for the satisficing problem is
%$\min\{O(|V|^2), \log{\frac{\mu}{|W|-v}}\}$ if $d\geq 2$ and $\min\{\O\Big(\frac{\log(\mu)}{d-1} + |V|^2\Big), \log{\frac{\mu}{|W|-v}}\}$ if $1<d<2$.
$\min\{O(|V|), \log{\frac{\mu}{|W|-v}}\}$ if $d\geq 2$ and $\min\{\O\Big(\frac{\log(\mu)}{d-1} + |V|\Big), \log{\frac{\mu}{|W|-v}}\}$ if $1<d<2$.
\end{theorem}

Observe that this bound is tight since the lower bounds from optimization apply here as well. The worst-case complexity can be completed using similar computations from \textsection~\ref{Sec:worstcasecomplexityVI}. Since, the number of iterations is identical to Theorem~\ref{thrm:vsquareiterations}, the worst-case complexity will be identical to Theorem~\ref{thrm:optimizationunitcost} and Theorem~\ref{thrm:optimizationbitcost}, showing no theoretical improvement. However, its implementations may terminate soon for threshold values far from the optimal but it will retain worst-case behavior for ones closer to the optimal. 
The catch is since the optimal cost is unknown apriori, this leads to a highly variable and {\em non-robust} performance. 
\section{Satisficing via Comparators}
\label{Sec:Satisfice}

Our second algorithm for satisficing is purely based on automata-methods. While this approach operates with integer discount factors only, it runs linearly in the size of the quantitative game. This is lower than the number of iterations required by VI, let alone the worst-case complexities of VI. 
This approach reduces satisficing to solving a safety or reachability game using comparator automata. 

The intuition is as follows: Given threshold value $v\in\Q$ and relation $\R$, let the satisficing problem be to ensure cost of plays relates to $v$ by $\R$. Then, a play $\rho$ is {\em winning for satisficing with $v$ and $\R$} if its cost sequence $A$ satisfies $\DSum{A}{d}$ $\R$ $v$, where $d>1$ is the discount factor. 
When $d$ is an integer and $v=0$, this simply checks if $A$ is in the safety/co-safety comparator, hence yielding the reduction. 

The caveat is the above applies to $v=0$ only. To overcome this, we extend the theory of comparators to permit arbitrary threshold values $v\in\Q$. We find that results from $v=0$ transcend to $v\in\Q$, and offer compact comparator constructions (\textsection~\ref{Sec:Comparator}). These new comparators are then used to reduce  satisficing to develop an efficient and scalable  algorithm (\textsection~\ref{Sec:Reduction}). Finally, to procure a well-rounded view of its performance, we conduct an empirical evaluation where we see this comparator-based approach outperform the VI approaches  \textsection~\ref{Sec:Implementation}.

\subsection{Foundations of comparator automata with threshold $v \in \Q$}
\label{Sec:Comparator}

This section extends the existing literature on comparators with threshold value $v=0$~\cite{BCVFoSSaCS18,BCVCAV18,BVCAV19} to permit non-zero thresholds. The properties we investigate are of safety/co-safety and $\omega$-regularity. We begin with formal definitions:

%For an integer upper bound $\mu>0$, discount factor $d >1$, and equality or inequality relation $\mathsf{R} \in \{<,>,\leq, \geq, =, \neq\}$the {\em DS comparison language with upper bound $\mu$, relation $\mathsf{R}$, and discount factor $d$} is the language of infinite words over the alphabet $\Sigma = \{-\mu, \dots, \mu\}$ that accepts  $A \in \Sigma^{\omega}$ iff $\DSum{A}{d}$ $\mathsf{R}$ $0$ holds~\cite{BVCAV19}. The {\em DS comparator automata  with upper bound $\mu$, relation $\mathsf{R}$, and discount factor $d$} is the automaton that accepts the DS comparison language with the same parameters~\cite{BCVFoSSaCS18}.Prior work has shown that a DS comparison language is a safety language if $\R \in \{\leq, \geq, =\}$, and a co-safety language if $\R \in \{<, >, \neq\}$~\cite{BVCAV19}.
%DS comparison language is said to be {\em $\omega$-regular} if its corresponding automaton is a B\"uchi automaton. It is known that DS comparison languages are $\omega$-regular iff the discount factor $d>1$ is an integer~\cite{BCVlmcs2019}.

\begin{definition}[Comparison language with threshold $v \in \Q$]
For an integer upper bound $\mu>0$, discount factor $d >1$, equality or inequality relation $\mathsf{R} \in \{<,>,\leq, \geq, =, \neq\}$, and a threshold value $v \in \Q$ 
the {\em comparison language with upper bound $\mu$, relation $\mathsf{R}$, discount factor $d$ and threshold value $v$} is a language of infinite words over the alphabet $\Sigma = \{-\mu, \dots, \mu\}$ that accepts  $A \in \Sigma^{\omega}$ iff $\DSum{A}{d}$ $\mathsf{R}$ $v$ holds. 
\end{definition}

\begin{definition}[Comparator automata with threshold $v \in \Q$]
For an integer upper bound $\mu>0$, discount factor $d >1$, equality or inequality relation $\mathsf{R} \in \{<,>,\leq, \geq, =, \neq\}$, and a threshold value $v \in \Q$ 
the {\em comparator automata with upper bound $\mu$, relation $\mathsf{R}$, discount factor $d$ and threshold value $v$} is an automaton that accepts the DS comparison language with upper bound $\mu$, relation $\mathsf{R}$, discount factor $d$ and threshold value $v$.
\end{definition}

% Prove that DS comparison language is safety/co-safety

\subsubsection{Safety and co-safety of comparison languages.}
\label{Sec:safetycosafety}

The primary observation is that to determine if $\DSum{A}{d}$ $\R$ $v$ holds, it should be sufficient to examine finite-length prefixes of $A$ since weights later on get heavily discounted. Thus,
\begin{theorem}
\label{Thrm:DSFull}
%[Safety and co-safety properties of DS-comparators]
Let $\mu>1$ be the integer upper bound.  
For arbitrary discount factor $d>1$ and threshold value $v\in\Q$
\begin{enumerate}
\item {\label{Item:Safety}}  Comparison languages are safety languages for relations $R \in \{\leq, \geq, =\}$.
\item {\label{Item:CoSafety}} Comparison language are co-safety languages for relations $R \in \{<, >, \neq\}$.
\end{enumerate}
\end{theorem}
\begin{proof}
The proof is identical to that for threshold value $v = 0$ from~\cite{BVCAV19}. \qed
\end{proof}

% Prove that DS comparator is \omega regualr for integer discount factors

\subsubsection{Regularity of comparison languages.}
\label{Sec:regularity}

Prior work on threshold value $v=0$ shows that a comparator is $\omega$-regular iff the discount factor is an integer~\cite{BCVlmcs2019}. We show the same result for arbitrary threshold values $v\in\Q$.

First of all, trivially, comparators with arbitrary threshold value are not $\omega$-regular for non-integer discount factors, since that already holds when $v = 0$.

The rest of this section proves $\omega$-regularity with arbitrary threshold values for integer discount factors. But first, let us introduce some notations: Since $v \in \Q$, w.l.o.g. we assume that the it has an $n$-length representation $v = v[0]v[1]\dots v[m](v[m+1]v[m+2]\dots v[n])^{\omega}$. By abuse of notation, we denote both the expression $ v[0]v[1]\dots v[m](v[m+1]v[m+2]\dots v[n])^{\omega}$ and the value $\DSum{v[0]v[1]\dots v[m](v[m+1]v[m+2]\dots v[n])^{\omega}}{d}$ by $v$.

We will construct a B\"uchi automaton for the comparison language $\L_\leq$ for
relation $\leq$, threshold value $v \in \Q$ and an integer discount factor.  
This is sufficient to prove $\omega$-regularity for all relations since B\"uchi automata are closed. 

From safety/co-safety of comparison languages, we argue it is sufficient to examine the discounted-sum of finite-length weight sequences to know if their infinite extensions will be in $\L_\leq$.
For instance, if the discounted-sum of a finite-length weight-sequence $W$ is {\em very large}, $W$ could be a bad-prefix of $\L_\leq$. 
%In this case, the discounted-sum of $W$ will be so large that for all infinite-length bounded extensions $Y$, $\DSum{W\cdot Y}{d} >v$. 
Similarly, if the discounted-sum of a finite-length weight-sequence $W$ is {\em very small} then for all of its infinite-length bounded extensions $Y$, $\DSum{W\cdot Y}{d} \leq v$. Thus, a mathematical characterization of {\em very large} and {\em very small} would formalize a criterion for membership of sequences in $\L_\leq$ based on their finite-prefixes. 

To this end, we use the concept of a {\em recoverable gap} (or gap value), which is a measure of distance of the discounted-sum of a finite-sequence from 0~\cite{boker2014exact}. 
The recoverable gap of a finite weight-sequences $W$ with discount factor $d$, denoted $\gap{W}$, is defined as follows: If $W = \varepsilon$ (the empty sequence), $\gap{\varepsilon} = 0$, and $\gap{W} = d^{|W|-1}\cdot \DSum{W}{d}$ otherwise. Then, Lemma~\ref{Lem:GapThreshold} formalizes {\em very large} and {\em very small} in Item~\ref{verylarge} and Item~\ref{verysmall}, respectively, w.r.t. recoverable gaps. As for notation, given a sequence $A$, let $A[\dots i]$ denote its $i$-length prefix:

\begin{lemma}
	\label{Lem:GapThreshold}
	Let $\mu>0$ be the integer upper bound,  $d >1$ be the discount factor.
	Let $v \in \Q$ be the threshold value such that 
	 $v = v[0]v[1]\dots v[m](v[m+1]v[m+2]\dots v[n])^{\omega}$.
	Let $W$ be a non-empty, bounded, finite-length weight-sequence.   
	\begin{enumerate}
	
	    \item \label{verylarge}  
        $\gap{W-\pre{v}{|W|}} >  \frac{1}{d}\cdot \DSum{\post{v}{|W|}}{d} + \frac{\mu}{d-1}$.
	    iff for all infinite-length, bounded extensions $Y$, $\DSum{W\cdot Y}{d} > v$ 
	    
        \item \label{verysmall} 
	    $\gap{W-\pre{v}{|W|}} \leq \frac{1}{d}\cdot \DSum{\post{v}{|W|}}{d} - \frac{\mu}{d-1}$
	    iff For all infinite-length, bounded extensions $Y$, $\DSum{W\cdot Y}{d}\leq v$

	\end{enumerate}
\end{lemma}
\begin{proof}
\sloppy
We present the proof of one direction of Item~\ref{verylarge}. The others follow similarly. 
Let $W$ be s.t. for every infinite-length, bounded extension $Y$,   $\DSum{W\cdot Y}{d}$ $> v$ holds. Then $ \DSum{W}{d} + \frac{1}{d^{|W|}}\cdot \DSum{Y}{d}$ $\geq$ $\DSum{\pre{v}{|W|}\cdot\post{v}{|W|}}{d}$  implies $\DSum{W}{d} - \DSum{\pre{v}{|W|}}{d}$ $>$ $\frac{1}{d^{|W|}}\cdot (\DSum{\post{v}{|W|}}{d} - \DSum{Y}{d})$ implies $\gap{W-\pre{v}{|W|}}$ $>$ $\frac{1}{d} (\DSum{\post{v}{|W|}}{d} + \frac{\mu\cdot d}{d-1})$.
\qed
\end{proof}

This segues into the state-space of the B\"uchi automaton. We define the state space so that state $s$ represents the gap value $s$. The idea is that all finite-length weight sequences with gap value $s$ will terminate in state $s$.
To assign transition between these states,  we observe that  gap value is defined inductively as follows: $\gap{\varepsilon} = 0$ and $\gap{W\cdot w} = d\cdot \gap{W} + w$, where  $w \in \{-\mu,\dots,\mu\}$. Thus there is a transition from state $s$ to state $t$ on $a \in \{-\mu,\dots,\mu\}$ if $t = d\cdot s + a$. 
Since $\gap{\varepsilon} = 0$,  state 0 is assigned to be the initial state.

The issue with this construction is it has infinite states. To limit that, we use Lemma~\ref{Lem:GapThreshold}. Since Item~\ref{verylarge} is a necessary and sufficient criteria for bad prefixes of safety language $\L_\leq$,  all states with value larger than Item~\ref{verylarge} are fused into one non-accepting sink. For the same reason, all states with gap value less than Item~\ref{verylarge} are accepting states. Due to Item~\ref{verysmall}, all states with value less than Item~\ref{verysmall} are fused into one accepting sink. 
Finally, since $d$ is an integer, gap values are integral. Thus, there are only finitely many states between Item~\ref{verysmall} and Item~\ref{verylarge}.

\begin{theorem} 
\label{thrm:regularWithVal}
Let $\mu>0$ be an integer upper bound,  $d >1$ an integer discount factor,  $\R $ an equality or inequality relation, and
$v \in \Q$ the threshold value with an $n$-length representation given by
 $v = v[0]v[1]\dots v[m](v[m+1]v[m+2]\dots v[n])^{\omega}$.
\begin{enumerate} 
    \item {\label{regular}}The DS comparator automata for $\mu,d,\R,v$ is $\omega$-regular iff $d$ is an integer. 
    \item {\label{safety}}For integer discount factors, the DS comparator is a safety or co-safety automaton with $\O(\frac{\mu\cdot n}{d-1})$ states.
\end{enumerate}
\end{theorem}

\begin{proof}
To prove Item~\ref{regular} we present the construction of an $\omega$-regular comparator automaton for integer upper bound $\mu>0$, integer discount factor $d>1$, inequality relation $\leq$, and threshold value $v \in \Q$ s.t. $v = v[0]v[1]\dots v[m](v[m+1]v[m+2]\dots v[n])^{\omega}$. 
, denoted by $\A=(\State, \Start, \Sigma, \delta, \Final)$ where:

For $i \in \{0,\dots,n\}$, let $\mathsf{U}_i = \frac{1}{d}\cdot \DSum{\post{v}{i}}{d} + \frac{\mu}{d-1}$ (Lemma~\ref{Lem:GapThreshold}, Item~\ref{verylarge})

For $i \in \{0,\dots,n\}$, let $\mathsf{L}_i = \frac{1}{d}\cdot \DSum{\post{v}{i}}{d} - \frac{\mu}{d-1}$ 
(Lemma~\ref{Lem:GapThreshold}, Item~\ref{verysmall})

\begin{itemize}
\item States $\State = \bigcup_{i=0}^n S_i \cup \{\mathsf{bad}, \mathsf{veryGood}\}$ where 
$S_i = \{(s, i) | s \in \{\floor{\mathsf{L}_i}+1, \dots , \floor{\mathsf{U}_i} \}\} $

\item Initial state $\Start = (0,0)$, Accepting states $\Final = S\setminus\{\mathsf{bad}\}$
\item Alphabet $\Sigma = \{-\mu, -\mu+1,\dots, \mu-1, \mu\}$
\item Transition function $\delta\subseteq \State \times \Sigma \rightarrow \State$ where $(s,a,t) \in \delta$ then:
	\begin{enumerate}
	\item \label{Trans:SelfLoop} If $s \in \{\mathsf{bad}, \mathsf{veryGood}\}$, then $t = s$ for all $a \in \Sigma$
	\item If $s$ is of the form $(p,i)$, and $a \in \Sigma$
		\begin{enumerate}
		\item \label{Trans:leq} If $d\cdot p+a-v[i] > \floor{\mathsf{U}_i}$, then $t = \mathsf{bad}$
		\item \label{Trans:geq} If $d\cdot p+a -v[i]\leq \floor{\mathsf{L}_i}$, then $t = \mathsf{veryGood}$
		\item \label{Trans:IntState} If $\floor{\mathsf{L}_i} <  d\cdot p+a -v[i] \leq \floor{\mathsf{U}_i}$, 
		\begin{enumerate}
		    \item If $i == n$, then $t = (d\cdot p+a - v[i], m+1)$
		    \item Else, $t = (d\cdot p+a - v[i], i+1)$
		\end{enumerate}
%		\item \label{Trans:IntState} If $\floor{-\thresh} +1 \leq d\cdot s+a \leq \floor{\thresh}$, then $t = d\cdot s+a$
		\end{enumerate}
	\end{enumerate}
\end{itemize}
We skip proof of correctness as it follows from the above discussion. Observe, $\A$ is  deterministic. It is a safety automaton as all non-accepting states are sinks.

To prove Item~\ref{safety},  observe that since the comparator for $\leq$ is a deterministic safety automaton, the comparator for $>$ is obtained by simply flipping the accepting and non-accepting states. This is a co-safety automaton of the same size. One can argue similarly for the remaining relations.
\qed
\end{proof}

\subsection{Satisficing via safety and reachability games}
\label{Sec:Reduction}

This section describes our comparator-based linear-time algorithm for satisficing for integer discount factors.

As described earlier, given discount factor $d>1$, a play is winning for satisficing with threshold value $v\in\Q$ and relation $\R$ if its cost sequence $A$ satisfies $\DSum{A}{d}$ $\R$ $v$. We now know from Theorem~\ref{thrm:regularWithVal}, that the winning condition for plays can be expressed as a safety or co-safety automaton for any $v \in\Q$ as long as the discount factor is an integer. Therefore, a {\em synchronized product} of the quantitative game with the safety or co-safety comparator denoting the winning condition completes the reduction to a safety or reachability game, respectively.

\begin{theorem}
\label{thrm:satisficingcomparator}
Let $G = (V, \init, E, \gamma)$ be a quantitative game,  $d >1$ the integer discount factor,  $\R$ the equality or inequality relation, and
$v \in \Q$ the threshold value with an $n$-length representation. Let $\mu>0$ be the maximum of absolute values of costs along transitions in $G$. Then,
\begin{enumerate}
    \item The satisficing problem reduces to solving a safety game if $\R \in \{\leq, \geq\}$ 
    \item The satisficing problem reduces to solving a reachability game if  $\R \in \{<,>\}$ 
    \item The satisficing problem is solved in $\O((|V| + |E|)\cdot \mu \cdot n)$ time. 
\end{enumerate}
\end{theorem}
\begin{proof}

The first two points use a standard synchronized product argument on the following formal reduction~\cite{colcombet2019universal}:
Let $G = (V  = V_0 \uplus V_1, \init, E, \gamma)$ be a quantitative game,  $d >1$ the integer discount factor,  $\R$ the equality or inequality relation, and
$v \in \Q$ the threshold value with an $n$-length representation. Let $\mu>0$ be the maximum of absolute values of costs along transitions in $G$.
Then, the first step is to construct the safety/co-safety comparator $\mathcal{A} = (\State, \Start, \Sigma, \delta, \Final)$ for $\mu$, $d$, $\R$ and $v$. 
The next is to synchronize the product of $G$ and $\A$ over weights to construct
the game $\GA = (W = W_0 \cup W_1, s_0 \times \mathsf{init}, \delta_W, \F_W)$, where
\begin{itemize}
    \item $W = V \times S$. In particular, $W_0 = V_0 \times S$ and $W_1 = V_1 \times S$. 
    Since $V_0$ and $V_1$ are disjoint, $W_0$ and $W_1$ are disjoint too. 
    
    \item Let $ s_0 \times \mathsf{init}$ be the initial state of $\GA$.
    
    \item Transition relation $\delta_W = W \times W$ is defined such that transition $((v,s), (v', s')) \in \delta_W$ synchronizes between transitions $(v, v') \in \delta$ and $(s, a, s') \in \delta_C$  if $a = \gamma((v, v'))$ is the cost of transition in $G$.

    \item  $\F_W = V \times \F$. 
    The game is a safety game if the comparator is a safety automaton and a reachability game if the comparator is a co-safety automaton. 
\end{itemize}

We need the size of $ \GA$ to analyze the worst-case complexity. Clearly, $\GA$ consists of $\O(|V|\cdot\mu\cdot n)$ states.
To establish the number of transitions in $\GA$, observe that every state  $(v, s)$ in $\GA$ has the same number of outgoing edges as state $v$ in $G$ because the  comparator $\A$ is deterministic. Since $\GA$ has  $\O(\mu\cdot n)$ copies of every state $v\in G$, there are a total of $\O(|E|\cdot \mu\cdot n)$ transitions in $\GA$. Since $\GA$ is either a safety or a reachability game, it is solved in linear-time to its size. Thus, the overall complexity is $\O((|V| + |E|)\cdot \mu \cdot n)$.
\qed
\end{proof}

With respect to the value $\mu$, the VI-based solutions are logarithmic in the worst case, while comparator-based solution is linear due to the size of the comparator. From a practical perspective, this may not be a limitation since weights along transitions can be scaled down. The parameter that cannot be altered is the size of the quantitative game. With respect to that, the comparator-based solution displays  clear superiority.  Finally, the comparator-based solution is affected by $n$, length of the representation of the threshold value while the VI-based solution does not. It is natural to assume that the value of $n$ is small.

\subsection{Implementation and Empirical Evaluation}
\label{Sec:Implementation}

The goal of the empirical analysis is to determine whether the practical performance of these algorithms resonate with our theoretical discoveries.

For an apples-to-apples comparison, we implement three algorithms:
(a) $\optimize$: Optimization via value-iteration,
(b)$\visatisfice$: Satisficing via value-iteration, and
(c). $\compSatisfice$: Satisficing via comparators. 
All tools have been implemented in \textsf{C++}. 
To avoid floating-point errors in $\optimize$ and $\visatisfice$, the tools invoke the open-source  \textsf{GMP} (\textsf{GNU} Multi-Precision)~\cite{gmp}.
Since all arithmetic operations in $\compSatisfice$ are integral only, it does not use \textsf{GMP}.

\begin{figure}[t]
\centering

\begin{minipage}{0.49\textwidth}
\centering
\includegraphics[width=\textwidth, trim = 1cm 1cm 1cm 2cm]{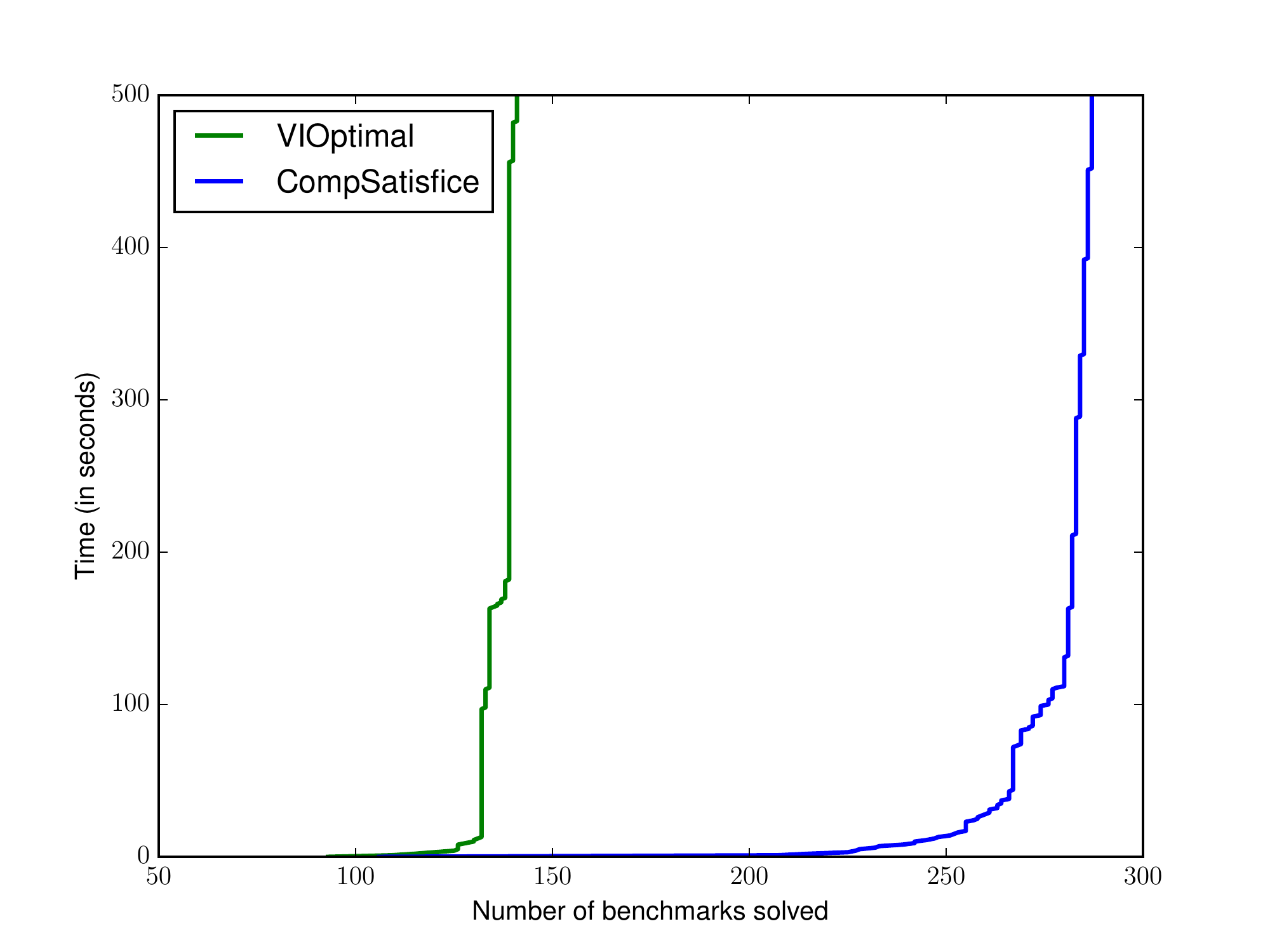}
\caption{Cactus plot.  $\mu=5, v=3$. Total benchmarks = 291}
\label{Fig:Cactus}
\end{minipage}
\hfill
\begin{minipage}{0.49\textwidth}
\centering
\includegraphics[width=\textwidth, trim = 1cm 1cm 1cm 2cm]{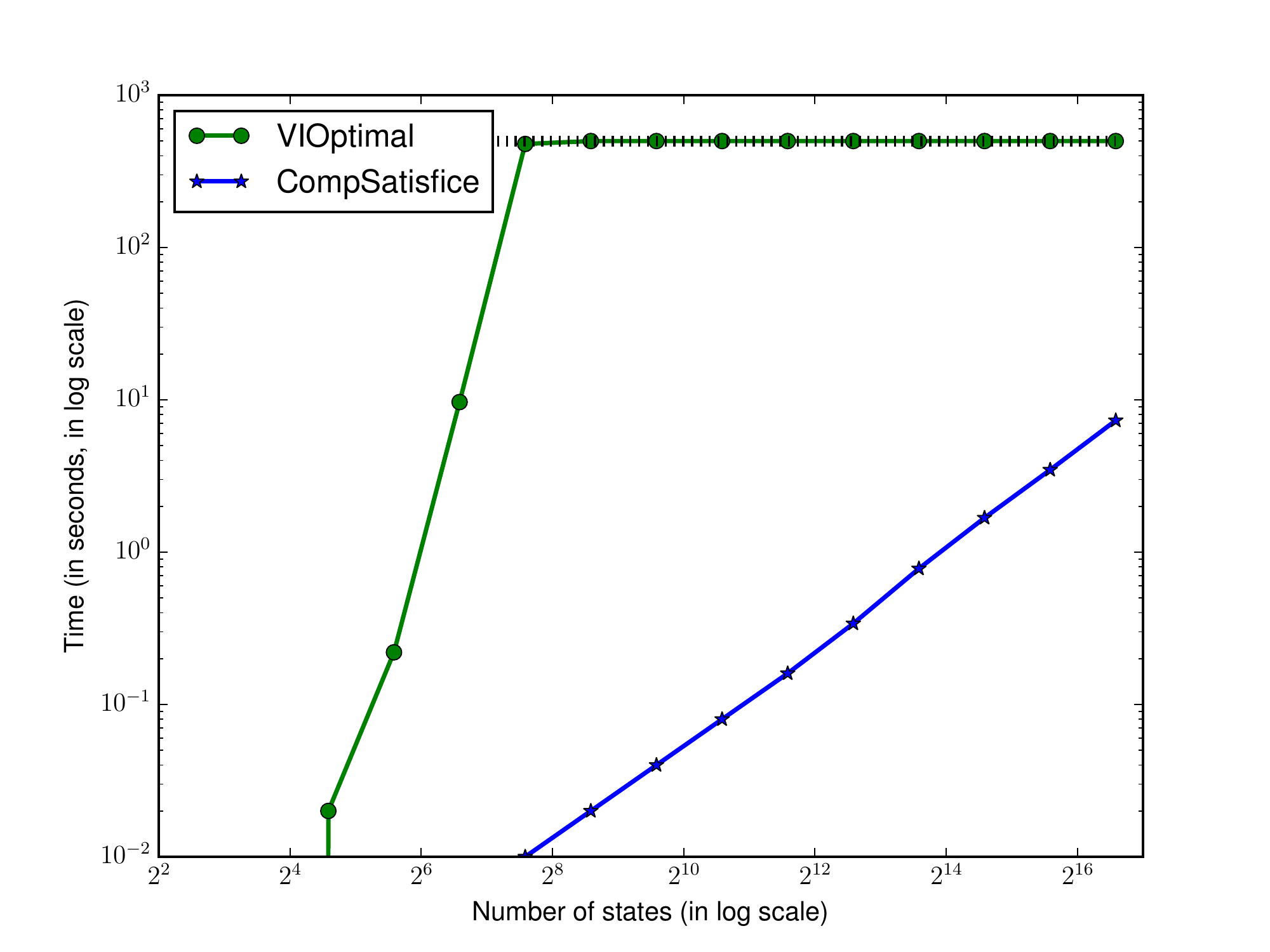}
\caption{Single counter scalable benchmark. $\mu=5, v=3$. Timeout = 500s.}
\label{Fig:Scale}
\end{minipage}

\end{figure}

To avoid completely randomized benchmarks,  we create $\sim$290 benchmarks from $\mathsf{LTL_f}$ benchmark suite~\cite{tabajara2019partition}.  
The state-of-the-art $\mathsf{LTL_f}$-to-automaton tool $\mathsf{Lisa}$~\cite{aaai2020} is used to convert $\mathsf{LTL_f}$ to (non-quantitative) graph games. Weights are randomly assigned to transitions. The number of states in our benchmarks range from 3 to 50000+. Discount factor $d = 2$, threshold $v \in [0-10]$.
Experiments were run on 8 CPU cores at 2.4GHz, 16GB RAM on a 64-bit Linux machine.

\subsubsection{Observations and Inferences}
 Overall, we see that $\visatisfice$ is efficient and scalable, and exhibits steady and predictable performance\footnote{Figures are best viewed online and in color}. 

\paragraph{$\compSatisfice$
outperforms $\optimize$} in both runtime and number of benchmarks solved, as shown in Fig~\ref{Fig:Cactus}. 
It is crucial to note that all benchmarks solved by $\optimize$ had fewer than 200 states. In contrast,  $\compSatisfice$ solves much larger benchmarks with 3-50000+ number of states.

To test scalability, we compared both tools on a set of scalable benchmarks. For integer parameter $i>0$, the $i$-th scalable benchmark has $3\cdot 2^i$ states.  
Fig~\ref{Fig:Scale} plots number-of-states to runtime in $\log$-$\log$ scale.
Therefore, the slope of the straight line will indicate the degree of polynomial (in practice). It shows us that  $\compSatisfice$ exhibits linear behavior (slope $\sim$1), whereas $\optimize$ is much more expensive (slope $>>1$) even in practice.

\paragraph{$\compSatisfice$ is more robust than $\visatisfice$.}

\begin{figure}[t]
\centering
\includegraphics[width=0.5\textwidth, trim=1cm 1cm 1cm 2.5cm]{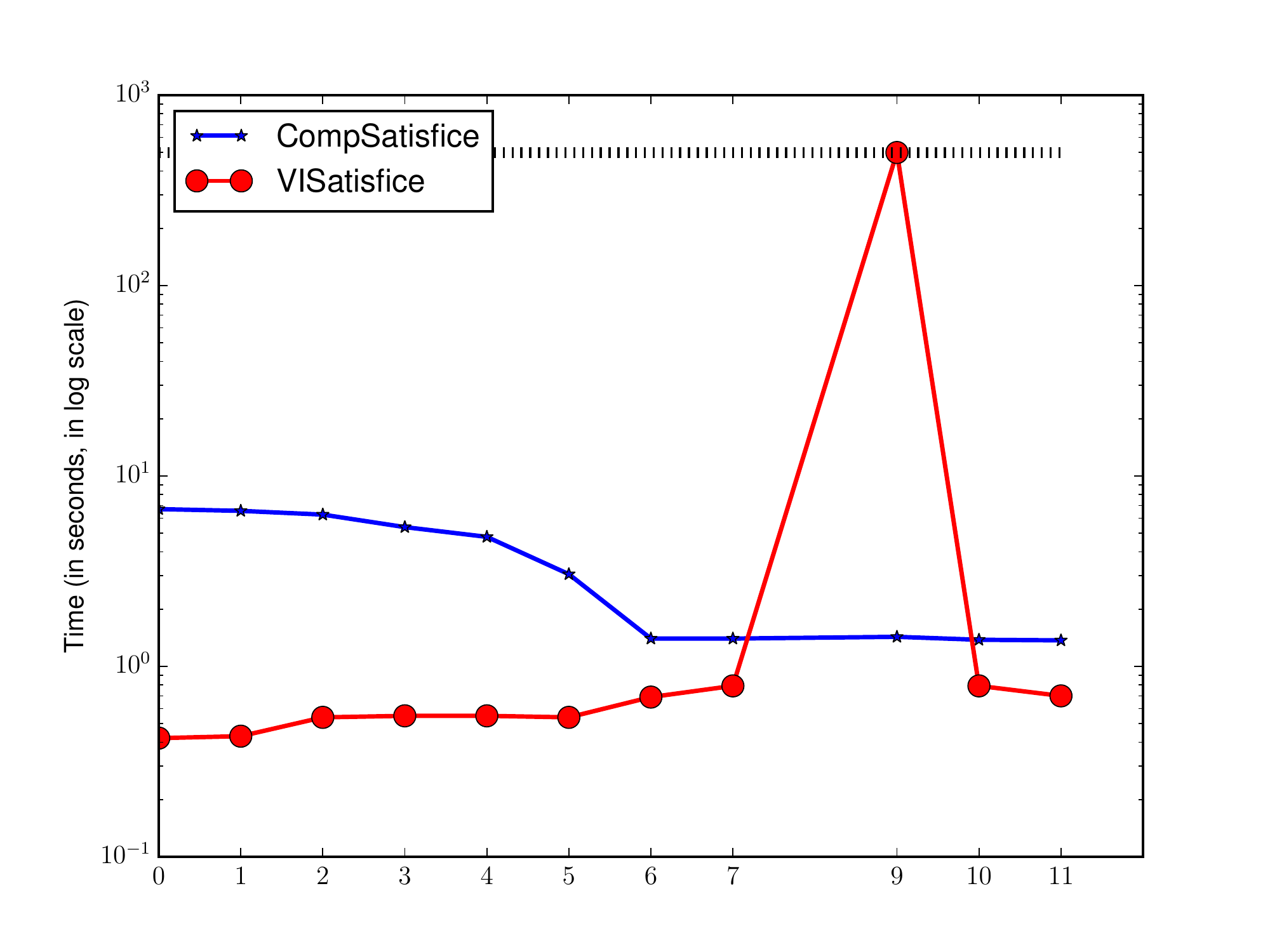}
%\vspace{0.5cm}
\caption{Robustness. Fix benchmark, vary $v$.  $\mu = 5$. Timeout = 500s.}
%\vspace{-0.5cm}
\label{Fig:Robust}
\end{figure}

We compare $\compSatisfice$ and $\visatisfice$  as the threshold value changes. This experiment is chosen due to Theorem~\ref{thrm:nonrobust} which proves that $\visatisfice$ is non-robust. 
As shown in Fig~\ref{Fig:Robust}, the variance in performance of $\visatisfice$ is very high. The appearance of peak close to the optimal value is an empirical demonstration of Theorem~\ref{thrm:nonrobust}.
On that other hand, $\compSatisfice$ stays steady in performance owning to its low complexity. 

\section{Adding Temporally Extended Goals}
\label{Sec:Temporalgoals}

Having witnessed algorithmic improvements of comparator-based satisficing over VI-based algorithms, we now shift focus to the question of applicability. While this section examines this with respect to the ability to extend to temporal goals, this discussion highlights a core strength of comparator-based reasoning in satisficing and shows its promise in a broader variety of problems.

The problem of extending optimal/satisficing solutions with a temporal goal is to determine whether there exists an optimal/satisficing solution that also satisfies a given temporal goal. Formally, given a quantitative game $G$, a labeling function $\L:V\rightarrow 2^{\ap}$ which assigns states $ V$ of $G$ to atomic propositions from the set $\ap$, and a temporal goal $\varphi$ over $\ap$, we say a {\em play $\rho = v_0 v_1\dots$ satisfies $\varphi$} if its proposition sequence given by $\L(v_0)\L(v_1)\dots$ satisfies the formula $\varphi$. Then to solve {\em optimization/satisficing with a temporal goal} is to determine if there exists a solutions that is optimal/satisficing and also satisfies the temporal goal along resulting plays. Prior work has proven that the optimization problem cannot be extended to temporal goals~\cite{chatterjee2017quantitative} unless the temporal goals are very simple safety properties~\cite{bernet2002permissive,wen2015correct}.
In contrast, our comparator-based solution for satisficing can naturally be extended to temporal goals, in fact to all $\omega$-regular properties, owing to its automata-based underpinnings, as shown below: 

\begin{theorem}
\label{thrm:withtemporalgoals}
Let $G$ a quantitative game with state set $V$, $\L:V\rightarrow 2^{\ap}$ be a labeling function over set of atomic propositions $\ap$, and  $\varphi$ be a temporal goal over $\ap$ and $\A_\varphi $ be its equivalent deterministic parity automaton. Let 
$d>1$ be an integer discount factor, $\mu$ be the maximum of the absolute values of costs along transitions, and $v\in Q$ be the threshold value with an $n$-length representation. 
Then, solving satisficing with temporal goals reduces to solving a parity game of size linear in $|V|$, $\mu$, $n$ and $|\A_\varphi|$.
\begin{comment}
\begin{itemize}
    \item Solving satisficing with temporal goals reduces to solving a parity game. The size of the parity game is linear in $|V|$, $\mu$, $n$, and double exponential in  $|\varphi|$.
    \item If $\varphi$ can be represented by a (deterministic) safety/co-safety automata $\A$, then solving the satisficing problem with temporal goals is $\O((|V|+|E|)\cdot \mu\cdot n \cdot |\A|)$, where $|\A| = 2^{2^{\O(|\varphi|)}}$. 
\end{itemize}
\end{comment}
\end{theorem}
\begin{proof}
The reduction involves two steps of  synchronized products. The first reduces the satisficing problem to a safety/reachability game while preserving the labelling function. The second synchronization product is between the safety/reachability game with the DPA $\A_\varphi$. These will synchronize on the atomic propositions in the labeling function and DPA transitions, respectively.
Therefore, resulting parity game will be linear in $|V|$, $\mu$ and $n$, and $|\A_\varphi|$.
\qed
\end{proof}

Broadly speaking, our ability to solve  satisficing via automata-based methods is a key feature as it propels a seamless integration of quantitative properties (threshold bounds) with qualitative properties, as both are grounded in automata-based methods. VI-based solutions are inhibited to do so since numerical methods are known to not combine well with automata-based methods which are so prominent with qualitative  reasoning~\cite{BCVCAV18,Kwi07}. This key feature could be exploited in several other problems to show further benefits of comparator-based satisficing over optimization and VI-based methods.

\section{Concluding remarks}

This work introduces the satisficing problem for 
quantitative games with the discounted-sum cost model.
When the discount factor is an integer, we present a comparator-based solution for satisficing, which exhibits algorithmic improvements -- better worst-case complexity and  efficient, scalable, and robust performance --  as well as broader applicability over traditional solutions based on numerical approaches for satisficing and optimization. Other technical contributions include the presentation of the missing proof of value-iteration for optimization and the extension of comparator automata to enable direct comparison to arbitrary threshold values as opposed to zero threshold value only.

An undercurrent of our comparator-based approach for satisficing is that it offers an automata-based replacement to traditional numerical methods. By doing so, it paves a way to combine quantitative and qualitative reasoning without compromising on theoretical guarantees or even performance. This motivates tackling more challenging problems in this area, such as more complex environments, variability in information availability, and their combinations.

\subsubsection{Acknowledgements.}
We thank anonymous reviewers for valuable inputs. This work is  supported in part by NSF grant 2030859 to
the CRA for the CIFellows Project, NSF grants IIS-1527668, CCF-1704883,
IIS-1830549, the ERC CoG 863818 (ForM-SMArt), and an award from the Maryland Procurement Office.

\bibliographystyle{abbrv}
\bibliography{myRef,refs}

\newpage
\appendix

\section{Complexity proof for VI Optimization}

\subsubsection{Lemma~\ref{lem:optimalplay}}
{\em
Let $l = l_1\cdot (l_2)^{\omega}$ represent  an integer cost sequence of a lasso, where $l_1$ and $l_2$ are the cost sequences of the head and loop of the lasso. Let $d = \frac{p}{q}$ be the discount factor. 
Then, $\DSum{l}{d}$ is a rational number with denominator at most ($p^{|l_2|} - q^{|l_2|})\cdot (p^{|l_1|})$.
}

\begin{proof}
The discounted sum of $l$ is given as follows:
\begin{align*}
    \DSum{l}{d} 
    = & \DSum{l_1}{d} + \frac{1}{d^{|l_1|}} \cdot(\DSum{(l_2)^\omega}{d})   \\
    = & \DSum{l_1}{d} + \frac{1}{d^{|l_1|}} \cdot \Big(\DSum{l_2}{d} + \frac{1}{d^{|l_2|}}\cdot \DSum{l_2}{d} + \frac{1}{d^{2\cdot|l_2|}}\cdot \DSum{l_2}{d}+ \dots\Big) \\
    & \text{Taking closed form expression of the term in the parenthesis, we get} \\
    = & \DSum{l_1}{d} + \frac{1}{d^{|l_1|}} \cdot \Big(\frac{d^{|l_2|}}{d^{|l_2|} -1 }\Big)\cdot \DSum{l_2}{d}  \\
    & \text{Let } l_2 = b_0b_1\dots b_{|l_2|-1} \text{ where } b_i \in \Z \\
    =  &\DSum{l_1}{d} + \frac{1}{d^{|l_1|}} \cdot \Big(\frac{d^{|l_2|}}{d^{|l_2|} -1 }\Big)\cdot \Big(b_0 + \frac{b_1}{d} + \dots + \frac{b_{|l_2|-1}}{d^{|l_2|-1}}\Big) \\
    =  &\DSum{l_1}{d} + \frac{1}{d^{|l_1|}} \cdot \Big(\frac{1}{d^{|l_2|} -1 }\Big)\cdot \Big(b_0\cdot d^{|l_2|} + {b_1}\cdot d^{|l_2|-1} + \dots + {b_{|l_2|-1}}{d}\Big) \\
    & \text{Expressing } d =\frac{p}{q} \text{, we get} \\
    = & \DSum{l_1}{d} + \frac{1}{d^{|l_1|}} \cdot \frac{q^{|l_2|}}{p^{|l_2|} - q^{|l_2|}}\cdot \Big(b_0 (\frac{p}{q})^{|l_2|} + \dots b_{|l_2|-1}\cdot \frac{p}{q}\Big) \\
    & \DSum{l_1}{d} + \frac{1}{d^{|l_1|}} \cdot \frac{1}{p^{|l_2|} - q^{|l_2|}}\cdot M, \text{ where } M \in \Z \\
    & \text{Expressing } d =\frac{p}{q} \text{ again, we get} \\
    = & \frac{1}{p^{|l_1|}} \cdot \frac{1}{p^{|l_2|} - q^{|l_2|}} \cdot N,  \text{ where } N \in \Z
\end{align*}
\end{proof}

\subsubsection{Theorem~\ref{thrm:vsquareiterations}}
{\em 
Let $G = (V, \init,E,\gamma)$ be a graph game. The number of iterations required by the value-iteration algorithm or the length of the finite-length game
to compute the optimal value $\opt$ is
\begin{enumerate}
    \item $\O(|V|)$ when discount factor $d\geq 2$, 
    \item $\O\Big(\frac{\log(\mu)}{d-1} + |V|\Big)$ when discount factor $1<d<2$.
    %\item $\O(|V|^2)$ when discount factor $d\geq 2$, 
    %\item $\O\Big(\frac{\log(\mu)}{d-1} + |V|^2\Big)$ when discount factor $1<d<2$.
\end{enumerate}
}

\begin{proof}

Recall, the task is to find a $k$ such that the interval identified by Lemma~\ref{lem:interval} is less than $\frac{1}{\mathsf{bound}_{\mathsf{diff}}}$.
Note that $\mathsf{bound}_{{\opt}} < \mathsf{bound}_{\mathsf{diff}}$. Therefore, $\frac{1}{\mathsf{bound_{diff}}} < \frac{1}{\mathsf{bound_{v}}}$. Hence, there can be only one rational value with denominator $\mathsf{bound}_{\mathsf{\opt}}$ or less in the small interval identified by the chosen $k$. Since the optimal value must also lie in this interval, the unique rational number with denominator $\mathsf{bound}_{\mathsf{\opt}}$ or less must be the optimal value. 
Let $k$ be such that the interval from Lemma~\ref{lem:interval} is less than $\frac{1}{\mathsf{bound}_{\mathsf{diff}}}$.
Then, 
\begin{align*}
     2\cdot\frac{\mu}{d-1\cdot d^{k-1}} \leq &  c\cdot \frac{1}{(p^{(|V|)} - q^{(|V|)})^2\cdot (p^{(2\cdot |V|)})} \text{ for some  } c>0  \\
    2\cdot\frac{\mu}{d-1\cdot d^{k-1}} \leq & c\cdot \frac{q^{4\cdot |V|}}{(p^{(|V|)} - q^{(|V|)})^2\cdot (p^{(2\cdot |V|)})} \text{ for some  } c>0  \\
    2\cdot\frac{\mu}{d-1\cdot d^{k-1}} 
     \leq & 
     c\cdot \frac{1}{(d^{(|V|)} - 1)^2\cdot (d^{(2\cdot |V|)})} \text{ for  } c>0  \\
    d-1\cdot d^{k-1} \geq &
     c' \cdot \mu \cdot (d^{(|V|)} - 1)^2\cdot (d^{(2\cdot |V|)}) \text{ for   } c'>0 \\
     \log(d-1) + (k-1)\cdot \log(d) \geq &
     c'' + \log(\mu) + 2\cdot \log(d^{(|V|)} - 1) +  {2\cdot |V|}\cdot\log(d) \text{ for   } c''>0 
\end{align*}

\begin{comment}
\begin{align*}
     2\cdot\frac{\mu}{d-1\cdot d^{k-1}} \leq &  c\cdot \frac{1}{(p^{(|V|^2)} - q^{(|V|^2)})\cdot (p^{(|V|^2)})} \text{ for some  } c>0  \\
    2\cdot\frac{\mu}{d-1\cdot d^{k-1}} \leq & c\cdot \frac{q^{2\cdot |V|^2}}{(p^{(|V|^2)} - q^{(|V|^2)})\cdot (p^{(|V|^2)})} \text{ for some  } c>0  \\
    2\cdot\frac{\mu}{d-1\cdot d^{k-1}} 
     \leq & 
     c\cdot \frac{1}{(d^{(|V|^2)} - 1)\cdot (d^{(|V|^2)})} \text{ for  } c>0  \\
    d-1\cdot d^{k-1} \geq &
     c' \cdot \mu \cdot (d^{(|V|^2)} - 1)\cdot (d^{(|V|^2)}) \text{ for   } c'>0 \\
     2\cdot\log(d-1) + k\cdot \log(d) \geq &
     c'' + \log(\mu) + \log(d^{(|V|^2)} - 1) +  {|V|^2}\cdot\log(d) \text{ for   } c''>0 
\end{align*}
\end{comment}

The following cases occur depending how large or small the values are:

\begin{description}
    
    \item[When $d\geq 2$:]  
    In this case, both $d$ and $d^{|V|}$ are large. Then, 
    \begin{align*}
    \log(d-1) + (k-1)\cdot \log(d) \geq &
     c'' + \log(\mu) + 2\cdot \log(d^{(|V|)} - 1) +  {2\cdot |V|}\cdot\log(d) \text{ for   } c''>0 \\
      \log(d) + (k-1)\cdot \log(d) \geq &
     c'' + \log(\mu) + 2\cdot |V|\cdot\log(d) +  {2\cdot |V|}\cdot\log(d) \text{ for   } c''>0 \\
     k & = \O(|V|)
    \end{align*}
    
    %\begin{align*}
    %    2\cdot\log(d-1) + k\cdot \log(d) & \geq 
    % c'' + \log(\mu) + \log(d^{(|V|^2)} - 1) +  {|V|^2}\cdot\log(d) \text{ for  } c''>0 \\
    % 2\cdot\log(d) + k\cdot \log(d) & \geq 
    % c'' + \log(\mu) + {(|V|^2)}\log(d) +  {|V|^2}\cdot\log(d) \text{ for  } c''>0 \\
    % k & = \O(|V|^2)
    %\end{align*}

    \item[When $d$ is small but $d^{|V|}$ is large:] In this case, $\log(d) \approx (d-1)$, and $\log(d-1) \approx 2-d$. Then,
     \begin{align*}
     \log(d-1) + (k-1)\cdot \log(d) \geq &
     c'' + \log(\mu) + 2\cdot \log(d^{(|V|)} - 1) +  {2\cdot |V|}\cdot\log(d) \text{ for   } c''>0 \\
    \log(d-1) + (k-1)\cdot \log(d) \geq &
     c'' + \log(\mu) + 2\cdot |V|\cdot\log(d) +  {2\cdot |V|}\cdot\log(d) \text{ for   } c''>0 \\
     2-d + (k-1)\cdot(d-1) \geq &
     c'' + \log(\mu) + 4\cdot |V|\cdot(d-1) \text{ for   } c''>0 \\
     k &= \O\Big(\frac{\log(\mu)}{d-1} + |V|\Big)
     \end{align*}
     %\begin{align*}
     %   2\cdot\log(d-1) + k\cdot \log(d) & \geq 
     %c'' + \log(\mu) + \log(d^{(|V|^2)} - 1) +  {|V|^2}\cdot\log(d) \text{ for  } c''>0 \\
     %2\cdot\log(d-1) + k\cdot \log(d) & \geq 
     %c'' + \log(\mu) + {(|V|^2)}\log(d) +  {|V|^2}\cdot\log(d) \text{ for  } c''>0 \\
     %2\cdot (2-d) + k\cdot (d-1) & \geq 
     %c'' + \log(\mu) + {(|V|^2)}(d-1) +  {|V|^2}\cdot (d-1) \text{ for  } c''>0 \\
     %k = \O\Big(\frac{\log(\mu)}{d-1} + |V|^2\Big)
    %\end{align*}
    \item[When both $d$ and $d^{|V|}$ are small:] Then, in addition to the approximations from the earlier case, $ \log(d^{|V|}-1) \approx (2-d^{|V|})$. So, 
    \begin{align*}
    \log(d-1) + (k-1)\cdot \log(d) \geq & c'' + \log(\mu) + 2\cdot \log(d^{(|V|)} - 1) +  {2\cdot |V|}\cdot\log(d) \text{ for   } c''>0 \\
    \log(d-1) + (k-1)\cdot \log(d) \geq &
     c'' + \log(\mu) + 2\cdot (2-d^{|V|}) +  {2\cdot |V|}\cdot\log(d) \text{ for   } c''>0 \\
     2-d + (k-1)\cdot (d-1) \geq &
     c'' + \log(\mu) + 2\cdot (2-d^{|V|}) +  {2\cdot |V|}\cdot(d-1) \text{ for   } c''>0 \\
     k &= \O\Big(\frac{\log(\mu)}{d-1} + |V|\Big)
    \end{align*}
    
   %  \begin{align*}
%        2\cdot\log(d-1) + k\cdot \log(d) & \geq 
%     c'' + \log(\mu) + \log(d^{(|V|^2)} - 1) +  {|V|^2}\cdot\log(d) \text{ for  } c''>0 \\
%     2\cdot\log(d-1) + k\cdot \log(d) & \geq 
%     c'' + \log(\mu) + 2-d^{(|V|^2)} +  {|V|^2}\cdot\log(d) \text{ for  } c''>0 \\
%     2\cdot (2-d) + k\cdot (d-1) & \geq 
%     c'' + \log(\mu) + 2-d^{(|V|^2)} +  {|V|^2}\cdot (d-1) \text{ for  } c''>0 \\
%     k = \O\Big(\frac{\log(\mu)}{d-1} + |V|^2\Big)
%    \end{align*}
    
\end{description}
\qed 
\end{proof}

\subsubsection{Concrete example to establish $\Omega(|V|)$ lower bound for number of iterations required by the value iteration algorithm}

Recall Fig~\ref{fig:lowerboundexample}, as presented here as well:
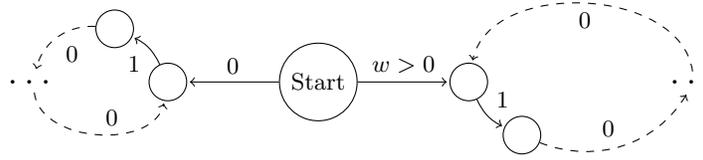
\begin{figure}[h]
\vspace{0cm}
    \centering
    
    \begin{tikzpicture}[shorten >=1pt,on grid,auto] 
   
    \node[state] (q_1) {\footnotesize{Start}};
    \node[state] (q_2)  [right of = q_1, node distance=2cm, minimum size=0.5cm] {};
    \node[state] (q_3)  [left of = q_1, node distance=2cm, minimum size=0.5cm] {};
    \node[] (q_4)  [right of = q_1, node distance=5cm] {\Large{\textellipsis}};
    \node[] (q_5) [left of=q_3, node distance=1.8cm]
    {\Large\textellipsis};
    \node[state] (q_6) [below right of = q_2, minimum size=0.5cm, node distance=1cm]{};
    \node[state] (q_7) [above left of = q_3, minimum size=0.5cm, node distance=1cm] {};
  
    \path[->] 
    (q_1) edge  node {$w>0$} (q_2)
          edge  node [above] {$0$} (q_3)
         
    (q_2) edge [bend right = 20] node {1} (q_6) 
    (q_6) edge [bend right = 40, dashed] node   {0} (q_4) 
    (q_3) edge [bend right = 20] node {1} (q_7) 
    (q_7) edge [bend right = 40, dashed] node   {0} (q_5) 
    (q_4) edge [bend right = 80, dashed] node  {0} (q_2)
    (q_5) edge [bend right = 80, dashed] node {0} (q_3);
    \end{tikzpicture}
    \caption{Sketch of game graph which requires $\Omega(|V|)$ iterations}
%\end{wrapfigure} 
\end{figure}

Let the left hand side loop have $4n$ edges, the right hand side of the loop have $2n$ edges, and $w = \frac{1}{d^{3n}} + \frac{1}{d^{7n}} + \cdots + \frac{1}{d^{m\cdot n -1}}$ such that $m\cdot n -1 = c\cdot n^2$
for a positive integer $c>0$.

One can show for a finite games of length  $(m\cdot n -1)$ or less, the optimal path arises from the loop to the right. But for games of length greater than  $(m\cdot n -1)$, the optimal path will be to due to the left hand side loop. 
\section{Complexity of VI under Bit-Cost model}

Under the bit-cost model, the cost of arithmetic operations depends on the size of the numerical values. Integers are represented in their bit-wise representation. Rational numbers $\frac{r}{s}$ are represented as a tuple of the bit-wise representation of integers $r$ and $s$. For two integers of length $n$ and $m$, the cost of their addition and multiplication is $O(m+n)$ and $O(m\cdot n)$, respectively.

To compute the cost of arithmetic in each iteration of the value-iteration algorithm, we define the cost of a transition $(v,w) \in E$ in the $k$-th iteration as
\begin{align*}
    \transwt_1(v,w)  = \gamma(v,w) \text{  and  }
    \transwt_k(v,w)  =  \gamma(v,w) + \frac{1}{d}\cdot\wt_{k-1}(v)  \text{ for } k >1
\end{align*}
Then, clearly, $\wt_k(v) = \max\{\transwt_k(v,w) | w \in vE\}$ if $v \in V_0$ and $\wt_k(v) = \min\{\transwt_k(v,w) | w \in vE\}$ if $v \in  V_1$. Since, we compute the cost of every transition in each iteration, it is crucial to analyze the size and cost of computing $\transwt$.

\begin{lemma}
\label{lem:transcostform}
Let $G$ be a quantiative graph game. Let $\mu>0$ be the maximum of absolute value of all costs along transitions. Let $d = \frac{p}{q}$ be the discount factor.
Then for all $(v,w )\in E$, for all $k>0$ 
\[
\transwt_k(v,w) = \frac{q^{k-1}\cdot n_1 + q^{k-2}p\cdot n_2 + \cdots + p^{k-1}n_{k}}{p^{k-1}}
\]
where $n_i \in \Z$ such that $|n_i| \leq \mu$  for all $ i \in \{1,\dots , k\}$.
\end{lemma}
Lemma~\ref{lem:transcostform} can be proven by induction on $k$.

\begin{lemma}
\label{lem:transcostcost}
Let $G$ be a quantiative graph game. Let $\mu>0$ be the maximum of absolute value of all costs along transitions. Let $d = \frac{p}{q}$ be the discount factor.
For all $(v,w) \in E$, for all $k>0$ the cost of computing $\transwt_k(v,w)$   in the $k$-th iteration is $\O(k\cdot \log p \cdot \max\{\log \mu, \log p\})$.
\end{lemma}
\begin{proof}
We compute the cost of computing $\transwt_k(v,w)$ given that optimal costs have been computed for the $(k-1)$-th iteration.
Recall,
\begin{align*}
    \transwt_k(v,w) & =  \gamma(v,w) + \frac{1}{d}\cdot\wt_{k-1}(v) 
    \text{ } =     \text{ } \gamma(v,w) + \frac{q}{p}\cdot\wt_{k-1}(v) \\
    & = \gamma(v,w) + \frac{q}{p}\cdot \frac{q^{k-2}\cdot n_1 + q^{k-3}p\cdot n_2 + \cdots + p^{k-2}n_{k-1}}{p^{k-2}}
\end{align*}
for some $n_i \in \Z$ such that $|n_i| \leq \mu$. Therefore, computation of $\transwt_k(v,w)$ involves four operations:
\begin{enumerate}
    \item Multiplication of $q$ with $(q^{k-2}\cdot n_1 + q^{k-3}p\cdot n_2 + \cdots + p^{k-2}n_{k-1})$. The later is bounded by $(k-1)\cdot \mu \cdot p^{k-1}$ since $|n_i| \leq \mu$ and $p>q$. The cost of this operation is $\O(\log ((k-1)\cdot \mu \cdot p^{k-1})\cdot \log(p)) = \O(((k-1)\cdot \log p + \log \mu + \log (k-1))\cdot(\log p))$.
    
    \item Multiplication of $p$ with $p^{k-2}$. Its cost is $\O((k-2)\cdot (\log p)^2)$.
    
    \item Multiplication of $p^{k-1}$ with $\gamma(v,w)$. Its cost is $\O((k-1)\cdot \log p \cdot \log \mu)$.
    
    \item Addition of $\gamma(v,w)\cdot p^{k-1}$ with $q\cdot (q^{k-2}\cdot n_1 + q^{k-3}p\cdot n_2 + \cdots + p^{k-2}n_{k-1})$. The cost is linear in their representations. 
\end{enumerate}
Therefore, the cost of computing $\transwt_k(v,w)$ is $\O(k\cdot \log p \cdot \max\{\log \mu, \log p\})$.
\qed
\end{proof}

Now, we can compute the cost of computing optimal costs in the $k$-th iteration from the $k-1$-th iteration. 
\begin{lemma}
\label{lem:costupdateiteation}
Let $G$ be a quantiative graph game. Let $\mu>0$ be the maximum of absolute value of all costs along transitions. Let $d = \frac{p}{q}$ be the discount factor.
The worst-case complexity of computing optimal costs in the $k$-th iteration from the $(k-1)$-th iteration is $\O(|E|\cdot k \cdot \log \mu \cdot \log p)$.
\end{lemma}
\begin{proof}
The update requires us to first compute the transition cost in the $k$-th iteration for every transition in the game. Lemma~\ref{lem:transcostcost} gives the cost of computing the transition cost of one transition. Therefore, the worst-case complexity of computing transition cost for all transitions is $\O(|E|\cdot k \cdot \log p \cdot \max\{\log \mu, \log p\})$.

To compute the optimal cost for each state, we are required to compute the maximum transition cost of all outgoing transitions from the state. Since the denominator is same, the maximum value can be computed via lexicographic comparison of the numerators on all transitions. Therefore, the cost of computing maximum for all states is $\O(|E|\cdot k \cdot \log \mu \cdot \log p)$.

Therefore, total cost of computing optimal costs in the $k$-th iteration from the $(k-1)$-th iteration is $\O(|E|\cdot k \cdot \log p \cdot \max\{\log \mu, \log p\})$.
\qed
\end{proof}

Finally, the worst-case complexity of computing the optimal value of the quantitative game under bit-cost model for arithmetic operations is as follows:

\subsubsection{Theorem~\ref{thrm:optimizationbitcost}.}
{\em 
Let $G = (V, \init, E, \gamma)$ be a quantitative graph game. 
Let $\mu>0$ be the maximum of absolute value of all costs along transitions. Let $d = \frac{p}{q}>1$ be the discount factor.
The worst-case complexity of computing the optimal value under bit-cost model for arithmetic operations is 
\begin{enumerate}
    \item $\O(|V|^2 \cdot |E| \cdot \log p \cdot \max\{\log \mu, \log p\})$ when $d \geq 2$, 
    \item $\O(\Big(\frac{\log(\mu)}{d-1} + |V|\Big)^2 \cdot |E| \cdot \log p \cdot \max\{\log \mu, \log p\})$ when $1<d < 2$.
\end{enumerate}
}
\begin{proof}
This is the sum of computing the optimal costs for all iterations.

When $d\geq 2$, it is sufficient to perform value iteration for $O(|V|)$ times (Theorem~\ref{thrm:vsquareiterations}).
So, the cost is $\O((1+2+3\cdot + |V|) \cdot |E| \cdot \log p \cdot \max\{\log \mu, \log p\})$. This expression simplifies to $\O(|V|^2 \cdot |E| \cdot \log p \cdot \max\{\log \mu, \log p\})$.

A similar computation solves the case for $1<d<2$.
\qed
\end{proof}

\section{Discounted-sum comparator construction}

\subsubsection{Theorem~\ref{Thrm:DSFull}}
{\em 
%[Safety and co-safety properties of DS-comparators]
Let $\mu>1$ be the upper bound.  
For arbitrary discount factor $d>1$ and threshold value $v$
\begin{enumerate}
\item {\label{Item:Safety}}  DS-comparison languages are safety languages for relations $R \in \{\leq, \geq, =\}$.
%non-strict inequalities and the equality relation
\item {\label{Item:CoSafety}} DS-comparison language are co-safety languages for relations $R \in \{<, >, \neq\}$.
\end{enumerate}
}
\begin{proof}
Due to duality of safety/co-safety languages, it is sufficient to show that DS-comparison language with $ \leq$ is a safety language.

Let us assume that DS-comparison language with $ \leq$  is not a safety language. Let $W$ be a weight-sequence in the complement of DS-comparison language with $ \leq$  such that it does not have a bad prefix. 

Since $W $ is in the complement of DS-comparison language with $ \leq$, $\DSum{W}{d} > v$. By assumption, every $i$-length prefix $W[i]$ of $W$ can be extended to a bounded weight-sequence $W[i]\cdot Y^i$ such that $\DSum{W[i]\cdot Y^i}{d} \leq v$. 

\sloppy
Note that  $\DSum{W}{d} = \DSum{\pre{W}{i}}{d} + \frac{1}{d^i} \cdot \DSum{\suf{W}{i}}{d}$, and $\DSum{\pre{W}{i}\cdot Y^i}{d} = \DSum{\pre{W}{i}}{d} + \frac{1}{d^i} \cdot \DSum{Y^i}{d}$. The contribution of tail sequences $\suf{W}{i}$ and $Y^i$ to the discounted-sum of $W$ and $\pre{W}{i}\cdot Y^i$, respectively diminishes exponentially as the value of $i$ increases. In addition, since and $W$ and  $\pre{W}{i}\cdot Y^i$ share a common $i$-length prefix $\pre{W}{i}$, their discounted-sum values must converge to each other.
The discounted sum of $W$ is fixed and greater than $v$, due to convergence there must be a $k\geq  0$ such that $\DSum{\pre{W}{k}\cdot Y^k}{d} > v$. Contradiction. 
Therefore, DS-comparison language with $ \leq$  is a safety language. 

The above intuition is formalized below:

Since $\DSum{W}{d} > v$ and $\DSum{\pre{W}{i}\cdot Y^i}{d} \leq v$, the difference $\DSum{W}{d} - \DSum{\pre{W}{i}\cdot Y^i}{d} > 0$. 

By expansion of each term, we get $\DSum{W}{d}- \DSum{\pre{W}{i}\cdot Y^i}{d} = \frac{1}{d^i}(\DSum{\suf{W}{i}}{d} - \DSum{Y^i}{d}) \leq  \frac{1}{d^i}\cdot(\mod{\DSum{\suf{W}{i}}{d}}  + \mod{\DSum{Y^i}{d}} )$. Since the maximum value of discounted-sum of sequences bounded by $\mu$ is $\frac{\mu\cdot d}{d-1}$, we also get that $\DSum{W}{d}- \DSum{\pre{W}{i}\cdot Y^i}{d} \leq 2\cdot \frac{1}{d^{i}} \mod{ \frac{\mu\cdot d}{d-1}} $. 

Putting it all together, for all $i\geq 0$ we get
\[
0< \DSum{W}{d} - \DSum{W[i]\cdot Y^i}{d} \leq 2\cdot \frac{1}{d^{i}} \mod{\frac{\mu\cdot d}{d-1}} 
\]
As $i \rightarrow \infty$, $2\cdot \mod{\frac{1}{d^{i-1}}\cdot \frac{\mu}{d-1}} \rightarrow 0$. So, $\lim_{i \rightarrow \infty} (\DSum{W}{d} - \DSum{W[i]\cdot Y^i}{d}) = 0$. Since $\DSum{W}{d}$ is fixed, $\lim_{i \rightarrow \infty}\DSum{W[i]\cdot Y^i}{d} = \DSum{W}{d}$.
%Therefore, $\lim_{i \rightarrow \infty}\DSum{W[i]\cdot Y^i}{d} > 0 $.

By definition of convergence, there exists an index $k\geq 0$ such that $\DSum{W[k]\cdot Y^k}{d}$ falls within the  $\frac{\mod{\DSum{W}{d}}}{2}$ neighborhood of $\DSum{W}{d}$. Finally since $\DSum{W}{d}> 0$, $\DSum{W[k]\cdot Y^k}{d} > 0$ as well. But this contradicts our assumption that for all $i \geq 0$, $\DSum{W[i]\cdot Y^i}{d} \leq 0$.

Therefore, DS-comparator with $ \leq$  is a safety comparator. 
\qed 
\end{proof}

\subsubsection{Lemma~\ref{Lem:GapThreshold}}
{\em 
	Let $\mu>0$ be the integer upper bound,  $d >1$ be an integer discount factor, and the relation $\R$ be the inequality $\leq$. 
	Let $v \in \Q$ be the threshold value such that 
	 $v = v[0]v[1]\dots v[m](v[m+1]v[m+2]\dots v[n])^{\omega}$.
	Let $W$ be a non-empty, bounded, finite weight-sequence.   
	Then, weight sequence $W$ is a bad-prefix of the DS comparison language with $\mu$, $d$, $\leq$ and $v$ iff $\gap{W-\pre{v}{|W|}} >  \frac{1}{d}\cdot \DSum{\post{v}{|W|}}{d} + \frac{\mu}{d-1}$.
}
\begin{proof}
\sloppy
Let $W$ be a bad prefix. Then for all infinite length, bounded weight sequence $Y$ we get that $\DSum{W\cdot Y}{d} > v \implies \DSum{W}{d} + \frac{1}{d^{|W|}}\cdot \DSum{Y}{d} \geq \DSum{\pre{v}{|W|}\cdot\post{v}{|W|}}{d} \implies \DSum{W}{d} - \DSum{\pre{v}{|W|}}{d} > \frac{1}{d^{|W|}}\cdot (\DSum{\post{v}{|W|}}{d} - \DSum{Y}{d})\implies \gap{W-\pre{v}{|W|}} > \frac{1}{d} (\DSum{\post{v}{|W|}}{d} + \frac{\mu\cdot d}{d-1})$.

Next, we prove that if a finite weight sequence $W$ is such that, the W is  a bad prefix. Let $Y$ be an arbitrary infinite but bounded weight sequence. Then $\DSum{W\cdot Y}{d} 
= \DSum{W}{d} + \frac{1}{d^{|W|}}\cdot \DSum{Y}{d} 
= \frac{1}{d^{|W|-1}}\cdot \gap{W} + \frac{1}{d^{|W|}}\cdot \DSum{Y}{d} 
= \frac{1}{d^{|W|-1}}\cdot \gap{W} + \frac{1}{d^{|W|}}\cdot \DSum{Y}{d} + \frac{1}{d^{|W|-1}}\cdot (\gap{\pre{v}{|W|}} - \gap{\pre{v}{|W|}})
$. By re-arrangement of terms we get that 
$\DSum{W\cdot Y}{d}
= \frac{1}{d^{|W|-1}}\cdot \gap{W-\pre{v}{|W|}} + \frac{1}{d^{|W|}}\cdot \DSum{Y}{d} + \frac{1}{d^{|W|-1}}\cdot \gap{\pre{v}{|W|}}$. Since 
 $ \gap{W - \pre{v}{|W|}} > \frac{1}{d}\cdot (\DSum{\post{v}{|W|}}{d} + \frac{\mu\cdot d}{d-1})$ holds, we get that 
 $\DSum{W\cdot Y}{d} 
 >  \frac{1}{d^{|W|}}\cdot (\DSum{\post{v}{|W|}}{d} + \frac{\mu\cdot d}{d-1}) +  \frac{1}{d^{|W|}}\cdot \DSum{Y}{d} + \frac{1}{d^{|W|-1}}\cdot \gap{\pre{v}{|W|}}$. Since minimal value of $\DSum{Y}{d} is \frac{-\mu\cdot d}{d-1}$, the inequality simplifies to  $\DSum{W\cdot Y}{d} > \frac{1}{d^{|W|-1}}\cdot \gap{\pre{v}{|W|}} + \frac{1}{d^{|W|}}\cdot \DSum{\post{v}{|W|}}{d}  \implies \DSum{W\cdot Y}{d} > \DSum{v}{d} = v$. Therefore, $W$ is a bad prefix. 
\qed
\end{proof}

\begin{lemma}
Let $\mu$ and $d>1 $ be the bound and discount-factor, resp.
Let $W$ be a non-empty, bounded, finite weight-sequence.   
Weight sequence $W$ is a very good-prefix of $\A^{\mu,d}_\leq$ iff $\gap{W-\pre{v}{|W|}} \leq \frac{1}{d}\cdot \DSum{\post{v}{|W|}}{d} - \frac{\mu}{d-1}$.

\end{lemma}
\begin{proof}
Let $W$ be a very good prefix. Then for all infinite, bounded sequences $Y$, we get that $\DSum{W\cdot Y}{d} \leq v \implies \DSum{W}{d} + \frac{1}{d^{|W|}} \cdot \DSum{Y}{d} \leq v$. By re-arrangement of terms, we get that $\gap{W-\pre{v}{|W|}} \leq \frac{1}{d}\cdot \DSum{\post{v}{|W|}}{d} - \frac{1}{d}\cdot \DSum{Y}{d}$. Since maximal value of $\DSum{Y}{d} = \frac{\mu\cdot d}{d-1}$, we get that $\gap{W-\pre{v}{|W|}} \leq \frac{1}{d}\cdot \DSum{\post{v}{|W|}}{d} - \frac{\mu}{d-1}$.

Next, we prove the converse. We know 
$\DSum{W\cdot Y}{d} 
= \DSum{W}{d} + \frac{1}{d^{|W|}}\cdot \DSum{Y}{d}
= \frac{1}{d^{|W|-1}}\cdot \gap{W} + \frac{1}{d^{|W|}}\cdot \DSum{Y}{d}
= \frac{1}{d^{|W|-1}}\cdot \gap{W} + \frac{1}{d^{|W|}}\cdot \DSum{Y}{d} +   \frac{1}{d^{|W|-1}}\cdot (\gap{\pre{v}{|W|}} - \gap{\pre{v}{|W|}})
$. 
By re-arrangement of terms we get that 

\sloppy
$\DSum{W\cdot Y}{d}
=\frac{1}{d^{|W|-1}}\cdot \gap{W-\pre{v}{|W|}} + \frac{1}{d^{|W|}}\cdot \DSum{Y}{d} + \frac{1}{d^{|W|-1}}\cdot\gap{\pre{v}{|W|}}$. From assumption we derive that 
$\DSum{W\cdot Y}{d} \leq \frac{1}{d^{|W|}}\cdot \DSum{\post{v}{|W|}}{d} - \frac{\mu}{d-1} + \frac{1}{d^{|W|}}\cdot\DSum{Y}{d}  + \DSum{\pre{v}{|W|}}{d}$. Since maximal value of $\DSum{Y}{d}$ is $ \frac{\mu}{d-1}$, we get that $\DSum{W\cdot Y}{d}\leq v$. Therefore, $W$ is a very good prefix.
\qed
\end{proof}

\end{document}